\documentclass[11pt]{article}
\usepackage{amssymb}
\usepackage{amsfonts}
\usepackage{amsmath}
\usepackage{amsthm}
\usepackage{latexsym}

\usepackage{mathtools}

\usepackage{xspace} 

\usepackage[pagebackref=true]{hyperref}
\hypersetup{
    unicode=false,          
    colorlinks=true,        
    linkcolor=red,          
    citecolor=blue,        
    filecolor=magenta,      
    urlcolor=cyan           
}

\usepackage[T1]{fontenc}
\usepackage[utf8x]{inputenc}

\usepackage{hyperref,cleveref}
\usepackage{amsmath}
\usepackage{amsthm}
\usepackage{amsfonts}
\usepackage{fullpage,appendix}
\usepackage{algorithmic}
\usepackage[ruled,vlined]{algorithm2e}
\usepackage{dsfont}
\usepackage{color}
\usepackage{tikz}
\usepackage{enumerate}
\usepackage{thm-restate}

\newcommand{\etal}{et al.\ }

\newcommand{\ignore}[1]{}
\definecolor{corlinks}{RGB}{64,128,128}
\definecolor{cormenu}{RGB}{0,37,94}
\definecolor{corurl}{RGB}{0,46,91}
\definecolor{darkgreen}{rgb}{0,0.5,0}

\newcommand{\E} {\mathbb{E}}

\newcommand{\B}{\mathbb{B}}

\newcommand{\C}{\mathbb{C}}

\newcommand{\poly}{\mathsf{poly}}

\newcommand{\F}{\mathbb{F}}

\renewcommand{\cal}[1]{\mathcal{#1}}

\newtheorem{fact}{Fact}[section]
\newtheorem{definition}[fact]{Definition}

\newtheorem{theorem}[fact]{Theorem}
\newtheorem{lemma}[fact]{Lemma}
\newtheorem{lem}[fact]{Lemma}

\newtheorem{proposition}[fact]{Proposition}
\newtheorem{problem}[fact]{Problem}
\newtheorem{claim}[fact]{Claim}
\newtheorem{property}[fact]{Property}

\newcommand{\Inote}[1]{}
\newcommand{\Pnote}[1]{}
\newcommand{\Mnote}[1]{}

\newcommand{\mmod}{~\mathrm{mod}~}

\newcommand{\bdd}[1]{{\sc RS-BDD}${(#1)}$}
\newcommand{\mss}[1]{{\sc MSS}${(#1)}$}
\newcommand{\symssum}[1]{{\sc SSS}${(#1)}$}

\usepackage{empheq}

\newcommand{\atomicsolverof}[2]{\textsc{AtomicSolver}{($#1$,$#2$)}\xspace}
\newcommand{\atomicsolverproc}[2]{\textsc{AtomicSolver}{($#1$,$#2$, $R_{#1, #2}$)}\xspace}
\newcommand{\atomicsolver}{\textsc{AtomicSolver}\xspace}
\newcommand{\generator}{\textsc{AuxiliaryVariableGenerator}\xspace}
\newcommand{\ra}{\rightarrow}

\title{NP-Hardness of Reed-Solomon Decoding,\\ and the Prouhet-Tarry-Escott Problem}
 \author{
Venkata Gandikota \thanks{Purdue University. {\tt vgandiko@purdue.edu}. Supported in part by a grant from the Purdue Research Foundation and by NSF CCF-1649515.}
\and
Badih Ghazi \thanks{Computer Science and Artificial Intelligence Laboratory, Massachusetts Institute of Technology, Cambridge MA 02139. Supported in part by NSF STC Award CCF 0939370 and NSF Awards CCF-1217423, CCF-1420956, CCF-1420692 and CCF-1217423.
  {\tt badih@mit.edu}.}
\and
Elena Grigorescu \thanks{Purdue University. {\tt elena-g@purdue.edu}. Supported in part by NSF CCF-1649515.}
}
\date{}

\begin{document}
\maketitle
\begin{abstract}

Establishing the complexity  of {\em Bounded Distance Decoding} for Reed-Solomon codes is a fundamental open problem in coding theory, explicitly asked by Guruswami and Vardy (IEEE Trans. Inf. Theory,  2005). The problem is motivated by the large current gap between the regime when it is NP-hard, and the regime when it is efficiently solvable (i.e., the Johnson radius).

We show the first  NP-hardness results for asymptotically smaller decoding radii than the maximum likelihood decoding radius of Guruswami and Vardy. Specifically, for Reed-Solomon codes of length $N$ and dimension $K=O(N)$, we show that it is NP-hard to decode more than $ N-K- c\frac{\log N}{\log\log N}$  errors  (with $c>0$ an absolute constant).  Moreover, we show that the problem is NP-hard under quasipolynomial-time reductions for an error amount $> N-K- c\log{N}$ (with $c>0$ an absolute constant). 

An alternative natural reformulation of the Bounded Distance Decoding problem for Reed-Solomon codes is as a {\em Polynomial Reconstruction} problem. In this view, our results show that it is NP-hard to decide whether there exists a degree $K$ polynomial passing through $K+  c\frac{\log N}{\log\log N}$ points from a given set of points  $(a_1, b_1), (a_2, b_2)\ldots, (a_N, b_N)$. Furthermore, it is NP-hard under quasipolynomial-time reductions to decide whether there is a degree $K$ polynomial passing through $K+c\log{N}$ many points.

These results follow from the NP-hardness of a generalization of the classical Subset Sum problem to higher moments, called {\em  Moments Subset Sum}, which has been a known open problem, and which may be of independent interest.

We further reveal a strong connection with the well-studied  Prouhet-Tarry-Escott problem in  Number Theory, which turns out to capture a main barrier in extending our techniques. We believe the  Prouhet-Tarry-Escott problem deserves further study in the theoretical computer science community.

\end{abstract}

\thispagestyle{empty}

\newpage
\tableofcontents
\newpage

\pagenumbering{arabic}

\section{Introduction}

Despite being a classical problem in the study of error-correcting codes, the computational complexity of decoding Reed-Solomon codes \cite{reed1960polynomial} in the presence of  large amounts of error is not fully understood. In the Bounded Distance Decoding problem, the goal is to recover a message corrupted by a bounded amount of error. Motivated by the large gap between the current efficient decoding regime, and the NP-hard regime for Reed-Solomon codes, we study the NP-hardness of Bounded Distance Decoding for asymptotically smaller error radii than  previously known. In this process, we unravel a strong connection with the Prouhet-Tarry-Escott, a famous  problem from number theory that has been studied for more than two centuries.

A  Reed-Solomon (RS) code of length $N$, dimension $K$, defined over a finite field $\F$, is the set of vectors (called {\em codewords}) corresponding to evaluations of low-degree univariate polynomials on a given set of evaluation points ${\cal D}=\{\alpha_1, \alpha_2, \dots, \alpha_N\}\subseteq \F$. Formally,  $RS_{{\cal D},K}=\{ \langle p(\alpha_1), \dots, p(\alpha_N)\rangle:  p\in \F[x] \mbox{ is a univariate polynomial of degree}   < K \}.$
The Hamming distance between $x, y \in \F^N$ is $\Delta(x, y):=|\{i \in [N]: ~ x_i \neq y_i\}|$.
In the {\em Bounded Distance Decoding (BDD) problem}, given a target vector $y\in \F^N$ and a distance parameter $\lambda$, the goal is to output $c\in {\cal C}$ such that $\Delta(c, y)\leq \lambda$. 

It is well-known that if the number of errors is $\lambda \le (N-K)/2$, there is a unique codeword within distance $\lambda$ from  the message, which can be found efficiently \cite{Peterson60, BerlekampWelch}. Further,  Sudan \cite{Sudan} and Guruswami and Sudan \cite{GuruSudan} show efficient decoding up to  $\lambda=N-\sqrt{NK}$ errors (the ``Johnson radius''), a setting in which the algorithm may output a small list of possible candidate messages. At the other extreme,  if the number of errors is at least  $ N-K$ (the covering radius),  finding one close codeword becomes trivial, amounting to interpolating a degree $K-1$ polynomial through $\leq K$  points. However, just below that radius, namely at  $N-K-1$ errors, the problem becomes NP-hard,  a celebrated result of Guruswami and Vardy \cite{GuruVardy}. The proof approach of \cite{GuruVardy} is only applicable to the Maximum Likelihood Decoding setting of $N-K-1$ errors, prompting the fundamental  problem  of understanding the complexity of BDD in the wide remaining range between $N-\sqrt{NK}$ and $N-K-1$:

\begin{quote}\cite{GuruVardy}{ {\em  ``It is an extremely interesting problem to show hardness of bounded distance decoding of Reed-Solomon codes for smaller decoding radius.''}}
\end{quote}

Some partial progress on improving the NP-hardness regime was shown in a  recent result by the same authors \cite{gandikota2015np} for  $N-K-2$ and $N-K-3$ errors.  The only other work addressing the hardness of decoding RS codes are due to Cheng and Wan \cite{ChengW07, WanCheng10} who show randomized reductions from the Discrete Log problem over finite fields, which is not believed to be NP-hard. 
 
In this work, we study the complexity of the decision version of  BDD, where the number of errors is parametrized by  $d\geq 0$, as formalized next:

\begin{itemize}
\item[]  {\bf Problem}  {\em Bounded Distance Decoding of Reed-Solomon codes with parameter $d$ (\bdd{d})}
\item[] {\bf Input} ${\cal D}=\{\alpha_1, \alpha_2, \dots, \alpha_N\}\subseteq \F$, where $\alpha_i \ne \alpha_j$ for all $i\ne j$,  target $y=(y_1, y_2, \dots, y_N)$, and integer $K<N$ 
\item[] {\bf Goal} Decide if there exists $p\in RS_{{\cal D}, K}$ such that $\Delta(y, p)\leq (N-K)-d$
\end{itemize}

We emphasize that the BDD problem above is in fact the basic and natural  Polynomial Reconstruction problem, where the input is a set of points ${\cal D}=\{(\alpha_1, y_1), (\alpha_2, y_2), \dots, (\alpha_N, y_N)\}\subseteq \F\times \F$, and the goal is to decide  if there exists a polynomial $p$ of degree $<K$ that passes through at least $K+d$ points in ${\cal D}$.

We state our main result in both forms.

\subsection{Contributions}

Our main technical contribution is the first NP-hardness result for BDD of RS codes, for a number of errors that is asymptotically smaller than  $N-K$, and its alternative view in terms of polynomial reconstruction.

\begin{theorem}
\label{thm-main}
There exists $c>0$, such that for every $1\leq d\leq c \cdot \frac{\log N}{\log\log N}$, the \bdd{d} problem for Reed-Solomon codes of length $N$, dimension $K=N/2-d+1$ and field size $|\mathbb{F}| = 2^{\poly(N)}$ is NP-hard. 
Furthermore, there exists $c>0$, such that for every $1\leq d\leq c \cdot \log N$,  \bdd{d} over fields of size $|\mathbb{F}| = 2^{N^{O(\log{N})}}$ does not have $N^{O(\log{N})}$-time algorithms unless NP has quasi-polynomial time algorithms. 

Equivalently, there exists $c>0$, such that for every $1\leq d\leq c \cdot \frac{\log N}{\log\log N}$, it is NP-hard to decide whether there exists a polynomial of degree $<K=N/2-d+1$ passing through $K+d$ many points from a given set   ${\cal D}=\{(\alpha_1, y_1), (\alpha_2, y_2), \dots, (\alpha_N, y_N)\}\subseteq \F\times \F$, with $|\mathbb{F}| = 2^{\poly(N)}$. Furthermore, there exists $c>0$, such that for every $1\leq d\leq c \cdot \log N$,  the same interpolation problem over fields of size $|\mathbb{F}| = 2^{N^{O(\log{N})}}$ does not have $N^{O(\log{N})}$-time algorithms unless NP has quasi-polynomial time algorithms. 
\end{theorem}

Our results significantly extend \cite{GuruVardy, gandikota2015np}, which only show NP-hardness for $d\in \{1, 2, 3\}$. 
As in \cite{GuruVardy, gandikota2015np}, we require the field size to be exponential in $N$.

The bulk of the proof of Theorem \ref{thm-main} is showing the NP-hardness of a natural generalization of the classic Subset Sum problem to higher moments,  that may be of independent interest. 

\begin{itemize}
\item[]  {\bf Problem}  {\em Moments Subset Sum with parameter $d$, over a field $\F$ (\mss{d})}
\item[] {\bf Input} Set $A\subseteq \F$ of size $|A|=N$, integer $k$, elements $m_1, m_2, \dots, m_d \in \F$
\item[] {\bf Goal} Decide if there exists $S\subseteq A$ such that $\sum_{s\in S} s^{\ell}=m_{\ell}$, for all $\ell\in[d]$, and $|S|=k$.
\end{itemize}

We note that the reduction from \mss{d} to \bdd{d} uses the equivalence between elementary symmetric polynomials and moments polynomials, when the field is of characteristic larger than $\Omega(d!)$(see Lemma~\ref{le:SMSS_RS_BDD_red} for a formal reduction.)

We point out that the Moments Subset Sum problem has natural analogs over continuous domains in the form of generalized moment problems and truncated moments problems,  which arise frequently in economics, operations research, statistics and probability \cite{lasserre2009moments}.

In this work, we prove NP-hardness of the Moments Subset Sum problem for large degrees.

\begin{theorem}\label{thm:mss}
There exists $c>0$, such that for every $1\leq d\leq c \cdot \frac{\log N}{\log\log N}$, the Moments Subset Sum problem \mss{d}  over prime fields of size $|\F|=2^{\poly(N)}$ is NP-hard. Furthermore, there exists $c>0$, such that for every $1\leq d\leq c \cdot \log N$, the Moments Subset Sum problem \mss{d} over fields of size $|\mathbb{F}| = 2^{N^{O(\log{N})}}$ does not have $N^{O(\log{N})}$-time algorithms unless NP has quasi-polynomial time algorithms. 
\end{theorem}

Furthermore, we reveal a connection with the famous  Prouhet-Tarry-Escott (PTE) problem in Diophantine Analysis, which is the main barrier for extending Theorem~\ref{thm:mss} and Theorem~\ref{thm-main} to  $d=\omega(\log N)$, as explained shortly.

The PTE problem \cite{prouhet1851memoire,dickson2013history,wright-prouhetssol} first appeared in letters between Euler and Goldbach in 1750-1751, and it is an important topic of study in classical  number theory (see, e.g., the textbooks of Hardy and Wright \cite{hardy-wright-book} and Hua \cite{hua-book}). It is also related to other classical problems in number theory, such as variants of the Waring problem and problems about minimizing the norm of cyclotomic polynomials, considered by Erd\"{o}s and Szekeres \cite{erdos1959product, BorweinIngalls}. 

In the Prouhet-Tarry-Escott  problem, given $k\geq 1$, the goal is to find disjoint sets of integers 
$\{ x_1, x_2, \dots, x_t\}$ and $\{y_1, y_2, \dots, y_t\}$ satisfying the system:
\begin{eqnarray*}
x_1+x_2+\dots+x_t &=& y_1+y_2+\dots+ y_t\\
x_1^2+x_2^2+\dots+x_t^2 &=& y_1^2+y_2^2+\dots+ y_t^2\\
&\dots&\\
x_1^k+x_2^k+\dots+x_t^k &=& y_1^k+y_2^k+\dots+ y_t^k.
\end{eqnarray*}
We call $t$ the size of the PTE solution. It turns out that the completeness proof of our reduction in Theorem~\ref{thm:mss} relies on \emph{explicit} solutions to this system for degree $k=d$ and of size $t = 2^k$. As explained next, despite significant efforts that have been devoted to constructing PTE solutions during the last 100 years, no explicit solutions of size $t = o(2^k)$ are known. This constitutes the main barrier to extending our Theorem~\ref{thm:mss} and Theorem~\ref{thm-main} to  $d=\omega(\log N)$.

The main open problem that has been tackled in the PTE literature is constructing solutions of small size $t$ compared to the degree $k$. It is relatively easy to show that $t\geq k+1$, and straightforward (yet non-constructive!) pigeon-hole counting arguments show the existence of solutions with $t=O(k^2)$. If we further impose the constraint that the system is not satisfied for degree $k+1$ (which is a necessary constraint for our purposes), then solutions of size $t=O(k^2 \log k)$ are known to exist \cite{hua-book}. However, these results are non-constructive, and the only general explicit solutions have size $t=O(2^k)$ (e.g., \cite{wright-prouhetssol,BorweinIngalls}). A special  class of solutions studied in the literature is for  $t=k+1$ (of minimum possible size). Currently there are known explicit parametric constructions of infinitely many minimum-size solutions for $k \leq 12$ (e.g., \cite{BorweinIngalls, BorweinLP03}), and finding such solutions often involves numerical simulations and extensive computer-aided searches \cite{BorweinLP03}.

From a computational point of view, an important open problem is to understand whether PTE solutions of size $O(k^2)$ (which are known to exist) can be \emph{efficiently constructed}, i.e., in time $\poly(k)$.

We identify the following generalization of the PTE problem as a current barrier to extending our results:

\begin{problem}\label{main-question} 
Given a field $\F$, integer $d$, and  $a, b\in \F$,  efficiently construct $x_1, \dots, x_t, y_1, \dots, y_t\in \F$, with $t=o(2^d)$, satisfying:
\begin{align*}
  x_1 + x_2 + \dots + x_t &= y_1 + y_2 + \dots + y_t \nonumber \\
  a^i + \displaystyle\sum\limits_{j=1}^t x_j^i &= b^i +\displaystyle\sum\limits_{j=1}^t y_j^i ~~~ \forall i \in \{2,\dots,d\}
\end{align*}
\end{problem}

We believe that this question is worth further study in the theoretical computer science community. In this work, we prove the following theorem, which is at the core of the completeness of our reduction.

\begin{theorem}
There is an explicit construction of solutions for \Cref{main-question} with $t= O(2^{d})$, and which can be computed in time $\poly(t)$.
\end{theorem}

In the next section, we outline the proof of Theorem~\ref{thm:mss}, and in the process, we explain how PTE solutions of degree $d$ naturally arise when studying the computational complexity of \mss{d}.

\subsection{Proof Overview} \label{subsec:pf_overv_tec}

To prove Theorem \ref{thm:mss}, we begin with the classical reduction from $1$-in-$3$-SAT to Subset-Sum, in which  one needs to construct a set of integers such that there is a subset whose sum equals a given target $m_1$, if and only if there is an assignment that satisfies exactly one literal of each clause of the $3$-SAT formula (we refer the reader to Section \ref{sec:reduction} for more details about this standard reduction). Extending this reduction so that the $2$nd moment also hits target $m_2$ raises immediate technical hurdles, since we have very little handle on the extra moment. In \cite{gandikota2015np}, the authors manage to handle a reduction for $2$nd and $3$rd moments via ad-hoc arguments and identities tailored to the degree-$2$ and degree-$3$ cases. The problem becomes much more complex as we need to ensure both completeness and soundness for a large number of moments.  In this work, we achieve such a reduction where the completeness will rely on explicit solutions to ``inhomogeneous PTE instances'' and the soundness will rely on a delicate balancing of the magnitudes of these explicit solutions. We now describe the details of this reduction.

For each $1$-in-$3$-SAT variable, we  create a collection of {\em explicit} auxiliary numbers which ``stabilize'' the contribution of this variable to all $i$-th moment equations with $2 \le i \le d$, 
 while having no net effect on the $1$st moment equation. Concretely, if $a$ and $b$ are the numbers corresponding to the two literals of the given variable, then we need to find numbers $x_1, \dots, x_t, y_1, \dots, y_t$ satisfying:
\begin{subequations}
\begin{empheq}{align}
  x_1 + x_2 + \dots + x_t &= y_1 + y_2 + \dots + y_t \nonumber \\
  a^i + \displaystyle\sum\limits_{j=1}^t x_j^i &= b^i +\displaystyle\sum\limits_{j=1}^t y_j^i ~~~ \forall i \in \{2,\dots,d\}\tag{\dag}\label{eq:mom_match}
\end{empheq}
\end{subequations}

Note that in order for the overall reduction to run in polynomial-time, the above auxiliary variables should be \emph{efficiently constructible}. Moreover, we observe that (\ref{eq:mom_match}) is an inhomogeneous PTE instance: for $a = b$, it reduces to a PTE instance of degree $d$. Of course, in our case $a$ and $b$ will not be equal, and (\ref{eq:mom_match}) is a more general system (and is hence harder to solve) than PTE instances. Nevertheless, as we will see shortly, solving (\ref{eq:mom_match}) can be essentially reduced to finding explicit PTE solutions of degrees $k \le d$.

In addition, we need to ensure that the added auxiliary numbers satisfy some ``bimodality''  property regarding their magnitudes, which would allow the recovery of a satisfying $1$-in-$3$-SAT assignment from any solution to the \mss{d} instance:

\begin{property}[Bimodality (informal)]\label{property:niceness_infor}
Every subset $S$ of the auxiliary variables is such that either $|\sum_{s\in S} s|$ is tiny, or  $|\sum_{s\in S} s|$ is huge.
\end{property}

We note that the existence of explicit and efficiently constructible solutions of small size $t=O(d)$ to system (\ref{eq:mom_match}) (and hence to a PTE system too)  would at least ensure the completeness of a reduction with $d=O(N)$. If soundness can also be ensured for such solutions, then our techniques would extend to radii closer to the Johnson Bound radius.

\paragraph{Overview of procedure for solving system (\ref{eq:mom_match})}

We build the variables $x_i$ and $y_i$ recursively, by reducing the construction for degree  $i$ to a solution to degree $i-1$. 
Towards this goal, we design a sub-procedure, called \atomicsolver, that takes as inputs  
an integer $i \in \{2,3, \dots ,d\}$, and a number $R_i$, and outputs $2^i$ rational\footnote{In our case, we can afford having \emph{rational} solutions to Equations~(\ref{eq:low_def_canc}) and~(\ref{eq:deg_i_match}). Note that this system is still a generalization of the PTE problem since we can always scale the rational solutions by their least common denominator to get a PTE solution of degree $i-1$.} numbers $\{ x_{i,j}, y_{i,j} \}_{j \in [2^{i-1}]}$ that satisfy a PTE  system of degree $i-1$, along with a non-homogeneous equation of degree $i$:

\begin{subequations}
\begin{empheq}{align}
  \displaystyle\sum\limits_{\ell=1}^{2^{i-1}} (x_{i,\ell}^j -  y_{i,\ell}^j) &= 0 ~~~ \forall ~2 \le j < i, \label{eq:low_def_canc}\\
  \displaystyle\sum\limits_{\ell=1}^{2^{i-1}} (x_{i,\ell}^i -  y_{i,\ell}^i) &= R_i\label{eq:deg_i_match}.
\end{empheq}
\end{subequations}

We can then run \atomicsolver sequentially on inputs $i \in \{2,\dots,d\}$ with the $R_i$ input corresponding to a ``residual'' term that accounts for the contributions to the degree-$i$ equation of the outputs of \atomicsolverof{j}{R_j} for all $2 \le j < i$, namely,

\begin{equation}
R_i = b^i - a^i + \displaystyle\sum\limits_{2 \le j <i} \displaystyle\sum\limits_{\ell=1}^{2^{j-1}} ( y_{j,\ell}^i -  x_{j,\ell}^i). \end{equation}

Note that the aim of the \atomicsolverof{i}{R_i} procedure is to satisfy the degree-$i$ equation (\ref{eq:deg_i_match}) without affecting the lower-degree equations (\ref{eq:low_def_canc}).

We then argue that the union $\cup_{2 \le i \le d}\{ x_{i,j}, y_{i,j} \}_{j \in [2^{i-1}]}$ of all output variables satisfies the polynomial constraints in (\ref{eq:mom_match}) with $t = \exp(d)$.

\paragraph{Specifics of the {\atomicsolver}}
 We next illustrate the \atomicsolver procedure by describing its operation in the particular case where $i=d=4$. In what follows, we drop ``$i=4$ subscripts'' and denote $R = R_4$, $x_{\ell} = x_{4,\ell}$ and $y_{\ell} = y_{4,\ell}$ for all $1 \le \ell \le 8$. Then, Equation~(\ref{eq:deg_i_match}) above that we need to satisfy becomes
\begin{equation}\label{eq:simp_deg_4}
\displaystyle\sum\limits_{\ell=1}^8 (x_{\ell}^4 - y_{\ell}^4) = R.
\end{equation}
First, we let $\alpha$ be a constant parameter (to be specified later on) and we set
\begin{subequations}
\begin{empheq}{align}
  x_1 - y_1 &= \alpha \label{eq:first_eq_alpha} \\ 
  y_2 - x_2 &= \alpha \label{eq:sec_eq_alpha}
\end{empheq}
\end{subequations}
Namely, in Equations~(\ref{eq:first_eq_alpha}) and (\ref{eq:sec_eq_alpha}), we ``\emph{couple}'' the ordered pairs $(x_1,y_1)$ and $(y_2,x_2)$ in the same way. Then, using Equations~(\ref{eq:first_eq_alpha}) and (\ref{eq:sec_eq_alpha}), we substitute $y_1 = x_1-\alpha$ and $x_2 = y_2-\alpha$, and the sum of the $\ell=1$ and $\ell=2$ terms in Equation~(\ref{eq:simp_deg_4}) can be written as
\begin{equation}\label{eq:cubic}
(x_1^4 - y_1^4) - (y_2^4-x_2^4) = p_{\alpha}(x_1) - p_{\alpha}(y_2)
\end{equation}
where $p_{\alpha}$ is a \emph{cubic} polynomial. If we set $x_1-y_2 = \beta$, then (\ref{eq:cubic}) further simplifies to
\begin{equation}\label{eq:p_to_q}
p_{\alpha}(x_1) - p_{\alpha}(y_2) = q_{\alpha,\beta}(x_1)
\end{equation}
where $q_{\alpha,\beta}$ is a \emph{quadratic} polynomial\footnote{Intuitively, we can think the LHS of (\ref{eq:p_to_q}) (along with the setting $x_1-y_2 = \beta$) as being a \emph{``derivative operator''}. This explains the fact that we are starting from a cubic polynomial $p_{\alpha}(\cdot)$ and getting a quadratic polynomial $q_{\alpha,\beta}(\cdot)$. This intuition was also used (twice) in (\ref{eq:cubic}), and will be again used in (\ref{eq:q_to_w}) and (\ref{eq:linear}) in order to reduce the degree further.}.

In the next step, we couple the ordered tuple $(y_3,x_3,y_4,x_4)$ in the same way that we have so far coupled the tuple $(x_1,y_1,x_2,y_2)$. The sum of the first four terms in the LHS of (\ref{eq:simp_deg_4}) then becomes
\begin{equation}\label{eq:quadratic}
\begin{split}
\displaystyle\sum\limits_{\ell=1}^4 (x_{\ell}^4 - y_{\ell}^4) &= (x_1^4 - y_1^4 + x_2^4 - y_2^4)-( y_3^4  - x_3^4 + y_4^4 - x_4^4)\\ 
&= q_{\alpha,\beta}(x_1) - q_{\alpha,\beta}(y_3).
\end{split}
\end{equation}
As before, we set $x_1 - y_3 = \gamma$ and (\ref{eq:quadratic}) further simplifies to
\begin{equation}\label{eq:q_to_w}
q_{\alpha,\beta}(x_1) - q_{\alpha,\beta}(y_3) = w_{\alpha,\beta,\gamma}(x_1)
\end{equation}
where $w_{\alpha,\beta,\gamma}(x_1)$ is a \emph{linear} polynomial in $x_1$. Finally, we couple the ordered tuple $(y_5,$ $x_5,$ $y_6,$ $ x_6, $ $y_7,x_7,y_8,x_8)$ in the same way that we have so far coupled the tuple $(x_1,y_1,x_2,y_2,x_3,y_3,x_4,y_4)$, and we obtain that the following equation is equivalent to Equation~(\ref{eq:simp_deg_4}) above:
\begin{equation}\label{eq:linear}
w_{\alpha,\beta,\gamma}(x_1) - w_{\alpha,\beta,\gamma}(y_5) = R. 
\end{equation}
Setting $x_1 - y_5 = \theta$, Equation~(\ref{eq:linear}) further simplifies to
\begin{equation}\label{eq:theta_h}
\theta \cdot h_{\alpha,\beta,\gamma} = R,
\end{equation}
where $h_{\alpha,\beta,\gamma}$ is the coefficient of $x_1$ in the linear polynomial $w_{\alpha,\beta,\gamma}(x_1)$. We conclude that to satisfy (\ref{eq:simp_deg_4}), it suffices to choose $\alpha$, $\beta, \gamma$ such that $h_{\alpha,\beta,\gamma} \neq 0$, and to then set $ \theta = R/h_{\alpha,\beta,\gamma}$.

It is easy to see that  there exist $\alpha, \beta, \gamma$ such that $h_{\gamma,\beta,\alpha} \neq 0$, and that  the above recursive coupling of the variables guarantees that (\ref{eq:low_def_canc}) is satisfied. The more difficult part will be to choose $\alpha, \beta, \gamma$ in a way that ensures the soundness of the reduction. This is briefly described next.

\paragraph{Bimodality of solutions}
In the above description of the particular case where $i = d = 4$, it can be seen that the produced solutions are $\{0, \pm 1\}$-linear combinations of $\{\alpha, \beta, \gamma, \theta\}$, which are required to satisfy (\ref{eq:theta_h}). It turns out that in this case $h_{\alpha,\beta,\gamma} = 24 \cdot \alpha \cdot \beta \cdot \gamma$, and so (\ref{eq:theta_h}) becomes
\begin{equation}\label{eq:deg_4_exp}
\theta \cdot \alpha \cdot \beta \cdot \gamma = \frac{R}{24}.
\end{equation}
So assuming we can upper bound $|R|,$\footnote{which we will do by inductively upper bounding $|R_i|$.} we would be able to set $\theta$ to a sufficiently large power of $10$ while letting $\alpha$, $\beta$ and $\gamma$ to have tiny absolute values and satisfy (\ref{eq:deg_4_exp}). Using the fact that the auxiliary $x_i$ and $y_i$ variables are set to $\{0, \pm 1\}$-linear combinations of $\{\alpha, \beta, \gamma, \theta\}$, this implies that the bimodality property is satisfied. In Section~\ref{sec:reduction}, we show that the bimodality property ensures that in any feasible solution to \mss{d}, the auxiliary variables should have no net contribution to the degree-$1$ moment equation (Proposition~\ref{prop:soundness}), which then implies the soundness of the reduction.

\paragraph{General finite fields} We remark that as described above, our solution works over the rational numbers, and, by scaling appropriately, over the integers. By taking the integer solution modulo a large prime $p$ (i.e., $p=2^{\poly(N)}$) the same arguments extend to $\F_p$. Moving to general finite fields $\F=\F_{p^{\ell}}$, we first observe that  system (\ref{eq:mom_match}) (and thus a PTE system too)  has non-constructive  solutions of size $O(d)$, which follows from the Weil bound (see Section \ref{subsec:existence}). Our reduction in the proof of Theorem~\ref{thm:mss} also extends to general fields $\F=\F_{p^{\ell}}$, where $p$ is a prime $p=\Omega(d!)$, and $\ell=\poly(N, d!)$.
The reduction now uses a representation of field elements in a polynomial basis  $\{1, \gamma, \gamma^2, \ldots, \gamma^{\ell-1}\} \subseteq \F$ , instead of decimal representations. 
See Section~\ref{subsec:general-fin-fields-reduction} for the changes that need to be made to the proof over the integers.

 \subsection{Related Work}

 A number of fundamental works address the polynomial reconstruction problem in various settings. In particular,  Goldreich \etal\  \cite{GoldreichRS00} show that that the polynomial reconstruction problem is NP-complete for univariate polynomials $p$ over large fields.  H{\aa}stad's celebrated results \cite{Hastad01}  imply NP-hardness for linear multivariate polynomials over finite fields. Gopalan \etal\  \cite{GopalanKS10} show NP-hardness for multivariate polynomials of larger degree, over the field $\F_2$.
   
We note that in general, the polynomial reconstruction problem does not require that the evaluation points are all distinct (i.e.,   $x_i\not=x_j$ whenever $i\not =j$). This  distinction is crucial to the previous results on polynomial reconstruction  (eg.  \cite{GoldreichRS00, GopalanKS10}).  It is this distinction that prevents those results from extending to the setting of Reed-Solomon codes, and to their multivariate generalization,  Reed-Muller codes.
   
On the algorithmic side, efficient algorithms for decoding of Reed-Solomon codes  and their variants are well-studied. As previously mentioned,  \cite{Sudan, GuruSudan} gave the first efficient algorithms in the list-decoding regime. Parvaresh and Vardy \cite{ParvareshV05} and Guruswami and Rudra \cite{GuruswamiR08} construct capacity achieving codes based on variants of RS codes. Koetter and Vardy \cite{KoetterV03a} propose soft decision decoders for RS codes. More recently,  Rudra and Wooters \cite{RudraWootters} prove polynomial list-bounds for random RS codes.

A related line of work is the study of BDD and of Maximum Likelihood Decoding in general codes, possibly under randomized reductions, and when an unlimited amount of preprocessing of the code is allowed. These problems have been extensively studied under diverse settings, e.g., \cite{Vardy97, AroraBSS97, DinurKRS03, DumerMS03, FeigeM04, Regev04, GuruVardy, Cheng08a}.

\section{Preliminaries}\label{sec:prelim}

We start by recalling the formal definition of the \mss{d} problem.
\begin{definition}[Moments Subset-Sum: \mss{d}]\label{le:defn_mss}
Given a set $A=\{a_1, \ldots, a_N\}$, $a_i \in \F$,   integer $k$, and $m_1, \ldots, m_d\in \F$, decide if there exists a subset $S\subseteq A$ of size $k$, satisfying $\sum_{a\in S} a^i=m_i$ for all $i\in [d]$. We call $k$ the size of the \mss{d} instance.
\end{definition}
\noindent We next recall the reduction from \mss{d} to RS-BDD$(d)$. 
\begin{lem}[\cite{gandikota2015np}]\label{le:SMSS_RS_BDD_red}
\mss{d} is polynomial-time reducible to RS-BDD$(d)$. Moreover, the reduction maps instances of \mss{d} on $N$ numbers and of size $k$ to Reed-Solomon codes of block length $N+1$ and of dimension $k-d+1$. The reduction holds over finite fields $\F$ of large characteristic. 

\end{lem}

~
\\

The reduction proceeds via \symssum{d}, a problem which is equivalent to \mss{d} over large fields. 

\begin{definition} [Symmetric Subset-Sum (\symssum{d})]
Given a set of $N$ distinct elements of $\mathbb{F}$,
$A = \{a_1, a_2, \dots, a_N\}$, integer $k$, and $E_1, E_2, \dots E_d \in \mathbb{F}$, 
decide if there exists a subset $S \subseteq A$ of size $k$, such that for every $i\in [d]$ the elementary symmetric sums of the elements of $S=\{s_1,\dots,s_k\} $ satisfy $E_i(S)=\sum_{1\leq j_1<j_2< \dots < j_i\leq k} s_{j_1}\dots s_{j_i}=E_i.$
\end{definition}

Given an instance $\langle A, k, E_1, E_2, \dots, E_d\rangle$ of \symssum{d}, we construct an instance $\langle {\cal D}, y, K \rangle$ of \bdd{d} 
such that there exists a Reed-Solomon codeword $p \in RS_{{\cal D}, K}$ with $\Delta(y, p)\leq (N-K)-d$
if and only if there is a solution to the given instance of \symssum{d}. 
\symssum{d} can be easily seen to be equivalent to \mss{d} over large prime finite fields $\F$ using Newton's identities \cite{Stanley99}, which will complete the proof of Lemma~\ref{le:SMSS_RS_BDD_red}. We note that this connection has been previously made (e.g. \cite{LiWan}).

\begin{lem}\label{le:SSS_RS_BDD_red}
\symssum{d} is polynomial-time reducible to \bdd{d}. 
\end{lem}

\begin{proof}
Given an instance $\langle A, k,E_1, E_2, \dots, E_d\rangle$ of \symssum{d}, we construct an instance $\langle {\cal D}, y, K \rangle$ of \bdd{d} such that there exists a Reed-Solomon codeword $p \in RS_{{\cal D}, K}$ with $\Delta(y, p)\leq (N-K)-d$ 
if and only if there is a solution to the given instance of \symssum{d}. Here, $A = \{a_1, a_2, \dots, a_N\}$ is a set of distinct, non-zero elements of $\mathbb{F}$ , $E_1, E_2, \dots, E_d \in \mathbb{F}$ and $k \in \mathbb{Z}$. 

Let $K = k-d+1$.
 Define the degree $d$ polynomial $f(x) := x^d - E_1 x^{d-1} + \dots + (-1)^{d-1}E_{d-1}x$. For each $a_i$ of $A$, define an element $ y_i \in \F$ as
$y_i =  -f(a_i) $. Define the target vector $y = (y_1, \cdots, y_N, (-1)^d E_d)$. 
The set ${\cal D}$ is then given by
${\cal D} = \{ a_1^{-1}, \cdots, a_N^{-1}, 0 \}$
Note that $\langle {\cal D}, y, K \rangle$ is an instance of \bdd{d} which can be constructed in 
polynomial time given the instance $\langle A, k, E_1, \dots, E_d \rangle$ of \symssum{d}.
 Let 
$D = \{(a_i^{-1}, y_i) \mbox{ for all } a_i \in A\}
\cup \{(0,(-1)^d E_d)\}$. Note that a Reed-Solomon codeword $p \in RS_{{\cal D}, K}$ at a distance $(N-K)-d$ from $y$ corresponds to a univariate polynomial $p(x)$ of degree at most $K-1$ which agrees with $D$ in $K+d$ points. 

Let $S$ be the solution to \symssum{d}. We now show that there exists a polynomial of degree at most $k-d ~(= K-1)$ which agrees with $D$ in at least $k+1 ~(= K+d)$ points. Define the following degree $k$ polynomial,
\begin{align*}
g(x) := \prod_{a_i \in S} (x - a_i) = c_0 + c_1x + \dots + c_{k-1} x^{k-1} + x^k
\end{align*}
The coefficients of this polynomial are the symmetric sums of the roots of $g(x)$. 
Therefore, $c_{k-d} = (-1)^d E_d, \dots, c_{k-2} = E_2$, and $c_{k-1} = -E_1$.
Now define,  
\begin{eqnarray*}
p(x) &=& (x^{k} g(1/x) - x^d f(1/x))/x^d\\
&=& c_0 x^{k-d} + c_1 x^{k-d-1} + \dots + c_{k-d}
\end{eqnarray*}
and note that $p(x)$ has degree $k-d$. We point out that $g(1/x)$ refers to the rational function obtained by replacing $x$ by $1/x$ in the polynomial $g(x)$. Also, the constant term of this polynomial is $c_{k-d} = (-1)^d E_d$.
Hence, $p(0) = (-1)^dE_d$ and since $g(a_i) = 0$, for all $a_i \in S$, it follows that 
$p(a_i^{-1}) = -f(a_i) = y_i \mbox{ for all } a_i \in S$.
Therefore, $p(x)$ agrees with $k+1$ points in $D$.

Conversely, we now show that if there is a polynomial $p(x)$ of degree at most $K-1~(=k-d)$ which agrees with $K+d~(=k+1)$ points in $D$, then there is a solution to \symssum{d}.
We first observe that if a degree $k-d$ polynomial passes through $k+1$ points of $D$, 
then it has to pass through $(0, (-1)^dE_d)$. To show this, assume $p(x)$ agrees with $k+1$ points of the form
$(a_i^{-1}, y_i) \in D$.
Let $g(x)$ be a degree $k$ polynomial defined as,
\begin{align*}
g(x) = x^{k-d}(p(1/x) + f(x)) 
\end{align*}
Therefore, if $p(x) = c_0 + c_1x + \dots + c_{k-d}x^{k-d}$, $g(x)$ can be written as
\begin{align*}
g(x) = ~& x^{k} + E_1 x^{k-1} + \dots + (-1)^{d-1} E_{d-1}x^{k-d+1} + c_0x^{k-d} + c_1 x^{k-d-1}  + \dots + c_{k-d}
\end{align*}
If  $p(a_i^{-1}) = y_i = -f(a_i)$  
for $k+1$ points, we have by definition that $g(a_i) = 0$ for these $k+1$~ $a_i's$.
This is a contradiction since $g(x)$ has degree at most $k$ and it cannot have $k+1$ roots.
Therefore, $p(0) = c_0 = (-1)^dE_d$. Also, $g(x)$ has $k$ roots which 
have their first $d$ symmetric sums  equal to $E_1, E_2, \dots, E_d$ respectively. Hence, there exists a 
solution to the given instance of \symssum{d}. 
\end{proof}

Given an instance $\langle A, k, B_1, \cdots, B_d\rangle$ of \mss{d}, we can construct an instance $\langle A, k, E_1, \cdots, E_d\rangle$ of \symssum{d} by setting 
\[
 E_{j}= \frac {1}{j!} \begin{vmatrix}B_{1}&1&0&\cdots \\B_{2}&B_{1}&2&0&\cdots \\\vdots &&\ddots &\ddots \\B_{j-1}&B_{j-2}&\cdots &B_{1}&j-1\\B_{j}&B_{j-1}&\cdots &B_{2}&B_{1}\end{vmatrix} ~~\text{ for every $j \in [d]$. } 
\]
The reduction from \mss{d} to \symssum{d} then follows from Newton's identities. Note that this reduction from from \mss{d} to \symssum{d} holds over finite fields $\F$ if $(j!)^{-1} \in \F$.

We will use the $1$-in-$3$-SAT problem in which we are given a $3$-SAT formula $\phi$ on $n$ variables and $m$ clauses and are asked to determine if there exists an assignment $z \in \{0,1\}^n$ satisfying exactly one literal in each clause. It is known that this problem is NP-hard even for $m = O(n)$ \cite{schaefer1978complexity}.
We will use $[n]$ to denote the set $\{1,2, \ldots, n\}$.
\section{Reduction from 1-in-3-SAT to \mss{d}}
\label{sec:reduction}

We start proving Theorem \ref{thm:mss} by describing the reduction from from 1-in-3-SAT to \mss{d} and its properties. 
Henceforth, we denote by $1^{\ell}$ the concatenation of $\ell$ ones, and we let $(1^{\ell})_{10}$ denote the positive integer whose decimal representation is $1^{\ell}$.

\paragraph{Subset Sum Reduction}\label{para:ssum}
 We start by recalling the reduction from $1$-in-$3$-SAT to Subset-Sum which will be used in our reduction to \mss{d}. In that reduction, each variable $(z_t, \overline{z_t})$, $t\in [n]$ is mapped to $2$ integers $a_t'$ (corresponding to $z_t$) and $b_t'$ (corresponding to $\overline{z_t}$). 
 The integers $a'_t$ and $b'_t$ and the target $B$  have the following decimal representation of  length-$(n+m)$: 
\begin{itemize}
\item The decimal representations of $a'_t$ and $b'_t$ consist of two parts: a variable region consisting of the leftmost $n$ digits and a clause region consisting of the (remaining) rightmost $m$ digits.
\item In the variable region, $a'_t$ and $b'_t$ have a $1$ at the $t$-th digit and $0$'s at the other digits. Denote that by $(a_t)^{'v}$.
\item In the clause region, for every $j \in [m]$, $a'_t$ (resp. $b'_t$) has a $1$ at the $j$th location if $z_t$ (resp. $\overline{z_t}$) appears in clause $j$, and a $0$ otherwise. We denote the clause part of $a'_t$ by $(a_t)^{'c}$.
\item We define $ a_t' = 10^{m} a_t^{'v} + a_t^{'c}$. We define $b_t'$ similarly. 
\item The target $B$ is set to the integer whose decimal representation is the all $1$'s, i.e., 
we set $B = 10^{m} (1^n)_{10} + (1^m)_{10}$. 
\end{itemize}
See  Figure~\ref{le:orig_red_fig} for an illustration of the decimal representations. This reduction to Subset-Sum is complete and sound.  
Indeed given a satisfying assignment to the 3-SAT formula $\phi(z)$, the subset $S = \{ a'_t \mid t \in [n], z_t = 1 \} \cup \{ b'_t \mid t \in [n], z_t = 0 \}$ is seen to satisfy that $\displaystyle\sum\limits_{s \in S} s = \displaystyle\sum_{\substack{ t \in [n]\\  z_t = 1}} a'_t + \displaystyle\sum_{\substack{ t \in [n]\\  z_t = 0}} b'_t = B$. Conversely, given a subset $S \subseteq \{ a'_t, b'_t \mid t \in [n] \}$ such that $\displaystyle\sum\limits_{s \in S} s = B$, a satisfying assignment to $\phi(z)$ is constructed from it by setting $z_i = 1$ if $a'_t \in S$ and $0$ otherwise. 
\noindent \begin{figure}[h]
\centering
\begin{tikzpicture}[scale=1.5]
\def \n {13}
\def \t {3}
\def \radius {1.15cm}
\def \margin {3} 
\def \margintwo {0.07} 
\def \radiussmall {0.7cm}
\def \radiuslarge {1.6cm}
\draw[<->,color=black] (-2,0) -- (0,0);
\draw[<->,color=black] (0,0) -- (2,0);
\node[fill=none, scale=0.8, blue] (n4) at (-1,-0.25) {variable region};
\node[fill=none, scale=0.8, blue] (n4) at (1,-0.25) {clause region};
\node[fill=none, scale=0.8, color=purple!75!black] (n4) at (-1,0.25) {$n$ digits};
\node[fill=none, scale=0.8, color=purple!75!black] (n4) at (1,0.25) {$m$ digits};
\node[fill=none, scale=0.8, darkgreen] (n4) at (-2.75,-0.75) {Target: $B = $};
\foreach \s in {1,...,2}
{
\node[fill=none, scale=0.8, color=darkgreen] (n4) at (-2.1+0.3*\s,-0.725) {$1$};
\node[fill=none, scale=0.8, color=darkgreen] (n4) at (-2.1+0.3*\s+0.1,-0.725) {$1$};
\node[fill=none, scale=0.8, color=darkgreen] (n4) at (-2.1+0.3*\s+0.2,-0.725) {$1$};
}
\foreach \s in {3,...,11}
{
\node[fill=none, scale=0.8, color=darkgreen] (n4) at (-2.1+0.3*\s,-0.725) {$\cdot$};
\node[fill=none, scale=0.8, color=darkgreen] (n4) at (-2.1+0.3*\s+0.1,-0.725) {$\cdot$};
\node[fill=none, scale=0.8, color=darkgreen] (n4) at (-2.1+0.3*\s+0.2,-0.725) {$\cdot$};
}
\foreach \s in {12,...,\n}
{
\node[fill=none, scale=0.8, color=darkgreen] (n4) at (-2.1+0.3*\s,-0.725) {$1$};
\node[fill=none, scale=0.8, color=darkgreen] (n4) at (-2.1+0.3*\s+0.1,-0.725) {$1$};
\node[fill=none, scale=0.8, color=darkgreen] (n4) at (-2.1+0.3*\s+0.2,-0.725) {$1$};
}
\end{tikzpicture}
\caption{Decimal representations in the original reduction from $1$-in-$3$-SAT to Subset-Sum. }\label{le:orig_red_fig}
\end{figure}
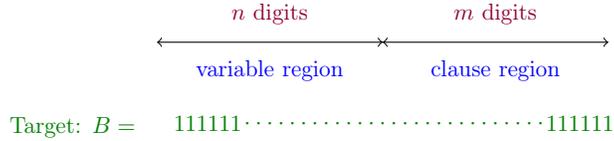

\paragraph{Our Reduction from 1-in-3-SAT to \mss{d}}
An instance of \mss{d} consists of a tuple $\langle A, k, B_1, \dots, B_d \rangle$. In this reduction, each variable $(z_t, \overline{z_t})$ is mapped to $2^{d+1}-2$ distinct rationals: $\{a_{t}\} \cup \{ x_{t,i} \mid i \in [2^{d}-2] \}$ (corresponding to $z_t$) and $\{b_{t}\} \cup \{ y_{t,i} \mid i \in [2^{d}-2] \}$ (corresponding to $\overline{z}_t$). Let $\{a'_t, b'_t: t \in [n]\}$ be the integers produced by the above reduction to Subset-Sum. We denote by $a_t^{'v}$ (resp. $a_t^{'c}$) the variable (resp. clause) region of $a'_t$.  Let $\nu$ be a natural number to be specified later on.  Define:
\begin{equation} \label{eqn:def-vars}
\begin{split}
a_t &:= 10^{\nu}(10^{m} a_t^{'v} + a_t^{'c}) ~\text{ and, } \\
b_t &:= 10^{\nu}(10^{m} b_t^{'v} + b_t^{'c}).
\end{split}
\end{equation} 
For each $t \in [n]$, we will explicitly construct two sets of $2^{d}-2$ {\em auxiliary variables}, $X_t = \{ x_{t, i} \mid i \in [2^{d}-2] \}$ and $Y_t = \{ y_{t, i} \mid i \in  [2^{d}-2] \}$ which satisfy the following properties:
\begin{enumerate}[Property (1):]
\item\label{p1} $\displaystyle\sum\limits_{x \in X_t} x = \displaystyle\sum\limits_{y \in Y_t} y = 0$.\\
\item\label{p2} $\displaystyle\sum\limits_{x \in X_t} x^k - \displaystyle\sum\limits_{y \in Y_t} y^k  = b_t^k - a_t^k \text{ for every } k \in \{2,\ldots,d\}$.
\item\label{p3} For any subset $S \subseteq  \bigcup\limits_{t\in[n]} (X_t \cup Y_t) $, either
$ \left| \displaystyle\sum\limits_{s \in S} s \right| > 10^{m+2n +\nu}$ or $ \left| \displaystyle\sum\limits_{s \in S} s \right| < 10^{\nu}$.
\item\label{p4} Every rational number of $\bigcup\limits_{t\in[n]} (X_t \cup Y_t)$ can be written as a fraction whose numerator and denominator are integers of magnitudes at most $10^{\poly(n,d!)}$. Moreover, $\left|\bigcup\limits_{t\in[n]} (X_t \cup Y_t) \right| = n \cdot (2^{d+1}-4).$
\end{enumerate}
Properties (\ref{p1}) and (\ref{p2}) will be used to ensure  completeness, Property (\ref{p3}) will be used to ensure soundness, and Property (\ref{p4}) will guarantee the polynomial running-time.
Constructing such auxiliary variables forms the crux of the reduction.

Define the set $A = \displaystyle\bigcup\limits_{t\in[n]} (\{ a_t \} \cup \{ b_t \} \cup X_t \cup Y_t)$. 
We will observe that $\lvert A \rvert =n(2^{d+1}-2)$ by showing that all the variables $\{ a_t \}, \{ b_t \}$ and those in $ X_t $ and $Y_t$ for ${t\in[n]}$  are distinct. 

Let $N = \lvert A \rvert =n(2^{d+1}-2), k = \frac N2$. The targets $B_1, \dots, B_d$ are defined as follows:
\begin{equation} \label{eqn:def-targets}
\begin{split}
B_1&:= 10^{\nu}(10^{m} (1^n)_{10} + (1^m)_{10}),\\
B_j &:= \displaystyle\sum\limits_{t = 1}^n a_t^j + \displaystyle\sum\limits_{t = 1}^n\displaystyle\sum\limits_{x \in X_t} x^j \text{   for every } j \in \{2, \dots, d \}.
\end{split}
\end{equation} 
Note that $a_{t}$ (and  $b_{t}$ and $B_1$, respectively) defined above are obtained by inserting $\nu$ zeros to the right of the decimal 
representation of $a_t'$ (resp. $b_t'$ and $B$).  Therefore, $a_t = 10^\nu \cdot a_t'$. Similarly, $b_t =10^\nu \cdot b_t' $ and $B_1 =10^\nu \cdot B$ (see Figure~\ref{le:k_2_red_fig} for a pictorial illustration).
The following fact is immediate from the definitions,  
\begin{fact} \label{claim:bound-at}
For any $x \in \{ a_t, b_t \mid t \in [n] \} \cup B_1$, we have 
\[ 10^\nu < \lvert x \rvert < 10^{m+n+\nu+1} \]
\end{fact}

\begin{figure*}[h]
\centering
\begin{tikzpicture}[scale=1.5]
\def \n {24}
\def \t {3}
\def \radius {1.15cm}
\def \margin {3} 
\def \margintwo {0.07} 
\def \radiussmall {0.7cm}
\def \radiuslarge {1.6cm}

\draw[<->,color=black] (-4,0) -- (-2,0);
\draw[<->,color=black] (-2,0) -- (0,0);
\draw[<->,color=black] (0,0) -- (2,0);
\draw[<->,color=black] (2,0) -- (4,0);
\node[fill=none, scale=0.8, color=purple!75!black] (n4) at (-3,-0.25) {large components region};
\node[fill=none, scale=0.8, color=purple!75!black] (n4) at (-1,-0.25) {variable region};
\node[fill=none, scale=0.8, color=purple!75!black] (n4) at (1,-0.25) {clause region};
\node[fill=none, scale=0.8, color=purple!75!black] (n4) at (3,-0.25) {tiny components region};

\node[fill=none, scale=0.8, blue] (n4) at (-1,0.25) {$n$ digits};
\node[fill=none, scale=0.8, blue] (n4) at (1,0.25) {$m$ digits};
\node[fill=none, scale=0.8, blue] (n4) at (3,0.25) {$\nu$ digits};

\node[fill=none, scale=0.8, color=darkgreen] (n4) at (-4.95,-0.75) {$1$st moment: $B_1 = $};

\foreach \s in {1,...,2}
{
\node[fill=none, scale=0.8, color=darkgreen] (n4) at (-4.25+0.3*\s,-0.725) {$0$};
\node[fill=none, scale=0.8, color=darkgreen] (n4) at (-4.25+0.3*\s+0.1,-0.725) {$0$};
\node[fill=none, scale=0.8, color=darkgreen] (n4) at (-4.25+0.3*\s+0.2,-0.725) {$0$};
}

\foreach \s in {3,...,5}
{
\node[fill=none, scale=0.8, color=darkgreen] (n4) at (-4.25+0.3*\s,-0.725) {$\cdot$};
\node[fill=none, scale=0.8, color=darkgreen] (n4) at (-4.25+0.3*\s+0.1,-0.725) {$\cdot$};
\node[fill=none, scale=0.8, color=darkgreen] (n4) at (-4.25+0.3*\s+0.2,-0.725) {$\cdot$};
}

\node[fill=none, scale=0.8, color=darkgreen] (n4) at (-4.25+0.3*5+0.3,-0.725) {$\cdot$};
\node[fill=none, scale=0.8, color=darkgreen] (n4) at (-4.25+0.3*5+0.4,-0.725) {$\cdot$};

\foreach \s in {6,...,6}
{
\node[fill=none, scale=0.8, color=darkgreen] (n4) at (-4.05+0.3*\s,-0.725) {$0$};
\node[fill=none, scale=0.8, color=darkgreen] (n4) at (-4.05+0.3*\s+0.1,-0.725) {$0$};
\node[fill=none, scale=0.8, color=darkgreen] (n4) at (-4.05+0.3*\s+0.2,-0.725) {$0$};
}

\foreach \s in {7,...,8}
{
\node[fill=none, scale=0.8, color=darkgreen] (n4) at (-3.95+0.3*\s,-0.725) {$1$};
\node[fill=none, scale=0.8, color=darkgreen] (n4) at (-3.95+0.3*\s+0.1,-0.725) {$1$};
\node[fill=none, scale=0.8, color=darkgreen] (n4) at (-3.95+0.3*\s+0.2,-0.725) {$1$};
}

\foreach \s in {9,...,11}
{
\node[fill=none, scale=0.8, color=darkgreen] (n4) at (-3.95+0.3*\s,-0.725) {$\cdot$};
\node[fill=none, scale=0.8, color=darkgreen] (n4) at (-3.95+0.3*\s+0.1,-0.725) {$\cdot$};
\node[fill=none, scale=0.8, color=darkgreen] (n4) at (-3.95+0.3*\s+0.2,-0.725) {$\cdot$};
}

\node[fill=none, scale=0.8, color=darkgreen] (n4) at (-3.95+0.3*11+0.3,-0.725) {$\cdot$};

\foreach \s in {12,...,12}
{
\node[fill=none, scale=0.8, color=darkgreen] (n4) at (-3.85+0.3*\s,-0.725) {$1$};
\node[fill=none, scale=0.8, color=darkgreen] (n4) at (-3.85+0.3*\s+0.1,-0.725) {$1$};
\node[fill=none, scale=0.8, color=darkgreen] (n4) at (-3.85+0.3*\s+0.2,-0.725) {$1$};
}

\foreach \s in {13,...,14}
{
\node[fill=none, scale=0.8, color=darkgreen] (n4) at (-3.75+0.3*\s,-0.725) {$1$};
\node[fill=none, scale=0.8, color=darkgreen] (n4) at (-3.75+0.3*\s+0.1,-0.725) {$1$};
\node[fill=none, scale=0.8, color=darkgreen] (n4) at (-3.75+0.3*\s+0.2,-0.725) {$1$};
}

\foreach \s in {15,...,17}
{
\node[fill=none, scale=0.8, color=darkgreen] (n4) at (-3.75+0.3*\s,-0.725) {$\cdot$};
\node[fill=none, scale=0.8, color=darkgreen] (n4) at (-3.75+0.3*\s+0.1,-0.725) {$\cdot$};
\node[fill=none, scale=0.8, color=darkgreen] (n4) at (-3.75+0.3*\s+0.2,-0.725) {$\cdot$};
}

\node[fill=none, scale=0.8, color=darkgreen] (n4) at (-3.75+0.3*17+0.3,-0.725) {$\cdot$};

\foreach \s in {18,...,18}
{
\node[fill=none, scale=0.8, color=darkgreen] (n4) at (-3.65+0.3*\s,-0.725) {$1$};
\node[fill=none, scale=0.8, color=darkgreen] (n4) at (-3.65+0.3*\s+0.1,-0.725) {$1$};
\node[fill=none, scale=0.8, color=darkgreen] (n4) at (-3.65+0.3*\s+0.2,-0.725) {$1$};
}

\foreach \s in {19,...,20}
{
\node[fill=none, scale=0.8, color=darkgreen] (n4) at (-3.55+0.3*\s,-0.725) {$0$};
\node[fill=none, scale=0.8, color=darkgreen] (n4) at (-3.55+0.3*\s+0.1,-0.725) {$0$};
\node[fill=none, scale=0.8, color=darkgreen] (n4) at (-3.55+0.3*\s+0.2,-0.725) {$0$};
}

\foreach \s in {21,...,23}
{
\node[fill=none, scale=0.8, color=darkgreen] (n4) at (-3.55+0.3*\s,-0.725) {$\cdot$};
\node[fill=none, scale=0.8, color=darkgreen] (n4) at (-3.55+0.3*\s+0.1,-0.725) {$\cdot$};
\node[fill=none, scale=0.8, color=darkgreen] (n4) at (-3.55+0.3*\s+0.2,-0.725) {$\cdot$};
}

\node[fill=none, scale=0.8, color=darkgreen] (n4) at (-3.55+0.3*23+0.1,-0.725) {$\cdot$};

\foreach \s in {24,...,\n}
{
\node[fill=none, scale=0.8, color=darkgreen] (n4) at (-3.45+0.3*\s,-0.725) {$0$};
\node[fill=none, scale=0.8, color=darkgreen] (n4) at (-3.45+0.3*\s+0.1,-0.725) {$0$};
\node[fill=none, scale=0.8, color=darkgreen] (n4) at (-3.45+0.3*\s+0.2,-0.725) {$0$};
}
\end{tikzpicture}
\caption{Decimal representations in the reduction from $1$-in-$3$-SAT to \mss{d}. The ``large components region'' only contains zeros in $\{a_{t},b_{t}: ~ t \in [n]\}$ but contains non-zeros in $\{|x_{t,i}|, |y_{t,i}|: ~ t\in[n], i \in [2^d-2]\}$.}\label{le:k_2_red_fig}
\end{figure*}

The following lemma is proved using Property (\ref{p4}) (and its proof appears in \Cref{sec:props}).

\begin{restatable}{lemma}{LemRuntime}
\label{lemma:runtime}
For any integer $d$, 
the total number of variables in the instance of \mss{d} is $N = n\cdot(2^{d+1}-2)$ and every variable has a $\poly(n, d!)$ digit representation in base $10$. 
\end{restatable}

In Section \ref{sec:gadget}, we will show how to construct variables satisfying Properties (\ref{p1}), (\ref{p2}), (\ref{p3}) and (\ref{p4}). 
The proof of Theorem~\ref{thm:mss} will follow from the next lemma and Lemma~\ref{lemma:runtime}.  
The proof of Theorem~\ref{thm-main} will then follow from Theorem~\ref{thm:mss} and Lemma~\ref{le:SMSS_RS_BDD_red}.

\begin{lemma}(Main)
\label{lemma:reduction}
There exists a satisfying assignment to a 3-SAT instance $\phi(z_1, \dots, z_n)$ if and only if there exists a subset $S\subseteq A$ of size $\lvert S \rvert = n(2^d-1)$ such that for every $k \in [d]$, 
\[ \displaystyle\sum\limits_{s \in S} s^k = B_k .\]
\end{lemma}

\begin{proof}[Proof of Theorem~\ref{thm:mss}]
Recall that $N=n(2^{d+1}-2)$, and so $|S|=|A|/2=N/2$. 
From Lemma~\ref{lemma:runtime} above, 
we know that every element constructed in the instance of \mss{d} 
has $\poly(n, d!)$ digit representation. 
Therefore, for $d = O(\log n/ \log \log n)$, the reduction runs in $\poly(n)$ time. 
 
 Let $c>0$ be a sufficiently small absolute constant. 
The NP-hardness of \mss{d} for $d < c \log{N}/\log\log{N}$ (under polynomial-time reductions) and for $d < c \log{N}$ (under quasipolynomial time reductions) over the field of rationals 
then follows from Lemma~\ref{lemma:reduction}.

By Lemma~\ref{lemma:runtime} above, we deduce the same hardness results for \mss{d} over prime fields of size $2^{\poly(N)}$. 
\end{proof}

We now prove Lemma~\ref{lemma:reduction}.
\begin{proof}[Proof of Lemma~\ref{lemma:reduction}] 
We start by proving the completeness of our reduction. 
We show that given a satisfying assignment $z$ to the 3-SAT instance $\phi(z_1, \dots, z_n)$, there exists a subset $S\subseteq A$ such that for every $k \in [d]$, 
\[ \displaystyle\sum\limits_{s \in S} s^k = B_k .\]

Consider the following subset $S$ of variables:
\[  S \triangleq \displaystyle\bigcup_{t \in [n], z_t = 1} \{ a_t \} \displaystyle\bigcup_{t \in [n], z_t = 1} X_t \displaystyle\bigcup_{t \in [n], z_t = 0} \{ b_t \} \displaystyle\bigcup_{t \in [n], z_t = 0} Y_t .\]
Note that $\lvert S \rvert = n(2^d-1) = \frac N2$ since the number of auxiliary variables included in $S$ corresponding to each $ t\in [n]$ is exactly $2^d - 2$.

For every $k \in [d]$, we have that
\begin{equation}
\label{eqn:completeness}
 \displaystyle\sum\limits_{s \in S} s^k =  \displaystyle\sum_{\substack{ t \in [n]\\  z_t = 1}} \left( a_t^k + \displaystyle\sum\limits_{x \in X_t} x^k \right) + 
 \displaystyle\sum_{\substack{ t \in [n]\\  z_t = 0}} \left( b_t^k  + \displaystyle\sum\limits_{y \in Y_t} y^k\right)
\end{equation}

By Property~(\ref{p2}) of the auxiliary variables, we have that for any $t \in [n]$ and any $k \in \{2, 3, \dots, d\}$, 
\[ \displaystyle\sum\limits_{x \in X_t} x^k - \displaystyle\sum\limits_{y \in Y_t} y^k  = b_t^k  - a_t^k.\]
Summing this equation over all $t \in [n]$, such that $z_t = 0$, we get 
\begin{equation} 
\label{eqn:props}
\displaystyle\sum_{\substack{ t \in [n]\\  z_t = 0}} \left( b_t^k +  \displaystyle\sum\limits_{y \in Y_t} y^k  \right) = 
\displaystyle\sum_{\substack{ t \in [n]\\  z_t = 0}} \left( a_t^k+ \displaystyle\sum\limits_{x \in X_t} x^k   \right)
\end{equation}

From \ref{eqn:completeness} and \ref{eqn:props}, we conclude that for every $k \in \{2,3,\dots,d\}$, 
\[  \displaystyle\sum\limits_{s \in S} s^k = \displaystyle\sum\limits_{t=1}^n \left( a_t^k + \displaystyle\sum\limits_{x \in X_t} x^k  \right) = B_k \]

For $k = 1$, Property~(\ref{p1}) implies that for every $t \in [n]$, $\displaystyle\sum\limits_{x \in X_t} x = 0$ and $\displaystyle\sum\limits_{y \in Y_t} y = 0$. Therefore, 
\begin{equation}
\label{eq:comp}
 \displaystyle\sum\limits_{s \in S} s = \displaystyle\sum_{\substack{ t \in [n]\\  z_t = 1}} a_t + \displaystyle\sum_{\substack{ t \in [n]\\  z_t = 0}} b_t 
 \end{equation}
Recall the variables $a_t', b_t'$ and $B$ from the Subset Sum reduction defined at the beginning of the proof. Note that 
$(\displaystyle\sum_{\substack{ t \in [n]\\  z_t = 1}} a'_t + \displaystyle\sum_{\substack{ t \in [n]\\  z_t = 0}} b'_t ) =  B$. Therefore,  we can rewrite Equation~(\ref{eq:comp}) as: 
\[  \displaystyle\sum\limits_{s \in S} s = 10^\nu \cdot \left( \displaystyle\sum_{\substack{ t \in [n]\\  z_t = 1}} a'_t + \displaystyle\sum_{\substack{ t \in [n]\\  z_t = 0}} b'_t \right) = 10^\nu \cdot B = B_1. \]

We now prove the soundness of our reduction. 
Let $S$ be a solution to the \mss{d} instance. That is, $S \subseteq A$ is such that 
$ \displaystyle\sum\limits_{s \in S} s^k = B_k$ for every $k \in [d]$. Proposition~\ref{prop:soundness} -- which is stated below -- shows that the auxiliary variables in $S$ should sum to $0$. Therefore, there exists a subset $S' \subseteq \{a_t, b_t \mid t \in [n] \}$ such that $\displaystyle\sum\limits_{s\in S'} s = B_1$. By definition of $a_t, b_t$ and $B_1$, it follows that there exists a subset of $\{a_t', b_t' \mid t \in [n] \}$ which sums to $B$, and the soundness of our reduction then follows from the soundness of the Subset Sum reduction.

\begin{proposition}
\label{prop:soundness}
Let $S \subseteq A$ be such that $ \displaystyle\sum\limits_{s \in S} s = B_1$. Let $D = \bigcup\limits_{t\in[n]} (X_t \cup Y_t)$ be the set of all the auxiliary variables. Then, 
\[ \displaystyle\sum\limits_{y \in S \cap D} y  = 0. \]
\end{proposition} 

\begin{proof}[Proof of Proposition~\ref{prop:soundness}]
Since $\displaystyle\sum\limits_{s \in S} s = B_1$, we have that 
\[ \displaystyle\sum\limits_{y \in S \cap D} y + \displaystyle\sum\limits_{s \in S \setminus D} s = B_1. \]
Note that $S \setminus D \subseteq \{ a_t, b_t \mid t \in [n] \}$. Since the $\nu$ least significant digits of $B_1$ and those of each element of $S\setminus D$ are all equal to $0$, either $\left| B_1 - \displaystyle\sum\limits_{s \in S \setminus D} s \right| = 0$ or $\left| B_1 - \displaystyle\sum\limits_{s \in S \setminus D} s \right| > 10^\nu.$ If $\left| B_1 - \displaystyle\sum\limits_{s \in S \setminus D} s \right| = 0$, then we are done. Henceforth, we assume that $\left| B_1 - \displaystyle\sum\limits_{s \in S \setminus D} s \right| > 10^\nu$. 
By Fact~\ref{claim:bound-at}, the elements of $S \setminus D$ as well as $B_1$ all have magnitudes at most $10^{m+n+\nu+1}$. Therefore, $\left| B_1 - \displaystyle\sum\limits_{s \in S \setminus D} s \right| \leq (2n+1) \cdot 10^{m + n + \nu +1} <  10^{m + 2n +\nu}$. On the other hand, by Property~(\ref{p3}) of the auxiliary variables, we know that either
$\left| \displaystyle\sum\limits_{y \in S \cap D} y \right| > 10^{m+2n+\nu}$ or 
$ \left| \displaystyle\sum\limits_{y \in S \cap D} y \right| < 10^{\nu}$. Since 
$\left| \displaystyle\sum\limits_{y \in S \cap D} y \right| = \left| B_1 - \displaystyle\sum\limits_{s \in S \setminus D} s \right|$, we get a contradiction. Therefore, 
$\displaystyle\sum\limits_{y \in S\cap D} y  = 0$.

\end{proof} 

\end{proof} 

\subsection{Constructing the auxiliary variables $X_t$, $Y_t$}\label{sec:gadget}

We now show how to construct the auxiliary variables, starting from the $a_t, b_t$ variables described before, for every $t\in [n]$. 
We do so in Algorithm \ref{algo:aux-vars}, the  \generator. For every $t \in [n]$, we construct $2(2^{d}-2)$ distinct auxiliary variables which satisfy the Properties \ref{p1}, \ref{p2}, \ref{p3} and \ref{p4} stated above.
The \generator outputs the union of 
 the variables generated in Algorithm~\ref{algo:gadget}, the \atomicsolver, using the recursive coupling idea  described  in Section~\ref{subsec:pf_overv_tec}. We use ${\bf 1}^{\ell}$ (and ${\bf 0}^{\ell}$) to denote a column vector of $\ell$ $1$'s ( $0$'s) respectively. For any vector $v$, let $v^T$ denote its transpose.

\begin{algorithm}[h]
  \caption{\generator:}
  \label{algo:aux-vars}
  \textbf{Input}: $\bigcup\limits_{t\in[n]} \{ a_t, b_t \}$\\
  \textbf{Output}: Sets of auxiliary variables $X_t, Y_t$ for every $t \in [n]$.\\
  \begin{algorithmic}[1]
  \FOR {$t \in [n]$}
	\STATE $X_t = \emptyset$
	\STATE $Y_t = \emptyset$
	\FOR {$i \in \{2, \dots, d\}$}
		\STATE $R_{t,i}= (b_t^i - a_t^i) + \displaystyle\sum\limits_{y \in Y_t} y^i - \displaystyle\sum\limits_{x \in X_t} x^i$
		\STATE Let $\left\{ x_{t, i, j} \mid j \in [2^{i-1}] \right\}  \bigcup  \left\{ y_{t, i, j} \mid j \in [2^{i-1}] \right\}  =  $ \atomicsolverproc{t}{i}
    		\STATE Let $X_t = X_t \bigcup \{ x_{t, i, j} \mid j \in [2^{i-1}] \} $ and  $Y_t = Y_t \bigcup \{ y_{t, i, j} \mid j \in [2^{i-1}] \} $
    \ENDFOR
\ENDFOR
\end{algorithmic}
\end{algorithm}

We now give the details of \atomicsolverproc{t}{i} for any $t\in [n]$ and $i \in \{2, 3, \dots, d\}$. 
Let $\nu = n^2$, and $M = m+ \nu + n + 1$. For every $t \in [n], i \in \{ 2, 3, \dots, d\}$ and $r \in [i]$, we define the functions
$f(t, i) :=  (i-1)! \cdot \nu_t$ and $g(t,i,r) := (t-1)d^2 + (i-1)i + r$, 
where $\nu_t$ is the $t^{th}$ prime integer greater than $n^4$. Note that $M = O(n^3)$ and $10^M > B_1$, by Fact~\ref{claim:bound-at}. We will use the fact that $\nu_t$ is much larger than $M$ later. Using the Prime Number Theorem \cite{Vshoup09}, it follows that the number of primes in the interval $[n^4, n^5] $ is larger than $n$, and thus $\nu_n < n^5$. Moreover, these $n$ primes can be found in deterministic polynomial time \cite{AKS04}. 

We will implement the recursive coupling idea of the {\atomicsolver} described  in Section~\ref{subsec:pf_overv_tec}, in terms of matrix algebra. For example, recall that in the first step of the variable coupling, we set $x_1-y_1=\alpha$, $y_2-x_2=\alpha$ and $x_1-y_2=\beta$. We can then express $x_1, x_2, y_1, y_2$ as a linear combination of $\alpha, \beta$, where we use the extra degree of freedom to choose $x_1=-x_2$ , as follows:
 $(x_1, x_2)^T=\frac12 \begin{bmatrix} 1 & 1 \\ -1 & -1 \end{bmatrix}\cdot (\alpha,  \beta)^T$, and $(y_1, y_2)^T=\frac 12 \begin{bmatrix} -1& 1 \\ 1 & -1 \end{bmatrix} \cdot (\alpha,  \beta)^T.$  
In general, the polynomial equations give rise to $2^{i}-1$ linear constraints on $2^i$ unknowns $(x_1, \cdots, x_{2^{i-1}}, y_1, \cdots, y_{2^{i-1}})$. The extra degree of freedom allows us to preserve the symmetry of the solution, which enables us to describe the algorithm and its analysis in a clean form.

\begin{algorithm}[h]
\caption{\atomicsolverproc{t}{i}: }
\label{algo:gadget}
\textbf{Input}: $i, t, R_{t,i}$  \\
\textbf{Output}: Set of auxiliary variables, $ \{x_{t, i, j} \mid j \in [2^{i-1}]\} \bigcup \{y_{t, i, j} \mid j \in [2^{i-1}]\} $\\
	\begin{algorithmic}[1]
	\STATE Let $\nu_t$ be the $t^{th}$ prime integer greater than $n^4$
	\STATE Let $f(t, i) =  (i-1)! \cdot \nu_t$
	\STATE Let $g(t,i,r) = (t-1)d^2 + (i-1)i + r$ for all $1 < r < i$
	\STATE $\alpha_{t, i, 1}= 10^{f(t,i)} $
	\STATE $\alpha_{t, i, r} = 10^{g(t,i,r)}$ for all $1 < r < i$
	\STATE $\alpha_{t, i, i} =  R_{t,i} / (i ! \prod\limits_{ r\in [i-1]} \alpha_{t,i,r})$ 
	\STATE $\alpha_{t, i} = [\alpha_{t, i, 1}, \dots, \alpha_{t, i, i}]^T$
	\IF{$i =2$}
		\STATE $A_2 = \begin{bmatrix} 1 & 1 \\ -1 & -1 \end{bmatrix}$ and  $B_2 = \begin{bmatrix} 1& -1 \\ -1 & 1 \end{bmatrix}$
	\ELSE
		\STATE $A_i = \begin{bmatrix} A_{i-1}& {\bf 1}^{2^{i-2}}  \\ B_{i-1} & -{\bf 1}^{2^{i-2}}  \end{bmatrix}$
and  $B_i = \begin{bmatrix} B_{i-1}& {\bf 1}^{2^{i-2}} \\ A_{i-1} & -{\bf 1}^{2^{i-2}} \end{bmatrix}$
	\ENDIF
	\STATE $[x_{t, i, 1}, \dots, x_{t, i, 2^{i-1}}]^T = \frac12 \cdot A_i \cdot \alpha_{t, i}$
	\STATE $[y_{t, i, 1}, \dots, y_{t, i, 2^{i-1}}]^T = \frac12 \cdot B_i \cdot \alpha_{t, i}$
	\STATE Return $ \{x_{t, i, j} \mid j \in [2^{i-1}]\} \bigcup \{y_{t, i, j} \mid j \in [2^{i-1}]\} $	
	\end{algorithmic}
\end{algorithm}

\begin{figure*}[h]
\centering
\begin{tikzpicture}[scale=1.5]
\def \n {24}
\def \t {3}
\def \radius {1.15cm}
\def \margin {3} 
\def \margintwo {0.07} 
\def \radiussmall {0.7cm}
\def \radiuslarge {1.6cm}

\draw[<->,color=black] (-4,0) -- (-2,0);
\draw[<->,color=black] (-2,0) -- (0,0);
\draw[<->,color=black] (0,0) -- (2,0);
\draw[<->,color=black] (2,0) -- (4,0);
\node[fill=none, scale=0.8, color=purple!75!black] (n4) at (-3,-0.25) {large components region};
\node[fill=none, scale=0.8, color=purple!75!black] (n4) at (-1,-0.25) {variable region};
\node[fill=none, scale=0.8, color=purple!75!black] (n4) at (1,-0.25) {clause region};
\node[fill=none, scale=0.8, color=green!75!black] (n4) at (-1,-0.50) {$n$-digits};
\node[fill=none, scale=0.8, color=green!75!black] (n4) at (1,-0.50) {$m$-digits};
\node[fill=none, scale=0.8, color=purple!75!black] (n4) at (3,-0.25) {tiny components region};
\node[fill=none, scale=0.8, color=green!75!black] (n4) at (3,-0.50) {$\nu$-digits};

\node[fill=none, scale=0.8, blue] (n4) at (-3,0.25) {$\alpha_{t, i, 1}$};
\node[fill=none, scale=0.8, blue] (n4) at (0,0.25) {$a_t, b_t, B_1$};
\node[fill=none, scale=0.8, blue] (n4) at (3,0.25) {$\alpha_{t, i, 2}, \dots, \alpha_{t, i, i}$};

\end{tikzpicture}
\caption{Relative distribution of $\alpha_{t, i, r}$ for any $i \in \{2, \cdots, d\}$ with respect to $a_t$, $b_t$ and $B_1$.}\label{le:k_3_red_fig}
\end{figure*}
\section{Verifying Properties \ref{p1}, \ref{p2}, \ref{p3}, \ref{p4} }\label{sec:props}
In this section, 
we prove that the variables generated by the \generator satisfy Properties~\ref{p1}, \ref{p2}, \ref{p3}, \ref{p4}. This is done via the following lemmas.

\begin{lemma} 
\label{lemma:condition12}
For every $t \in [n]$, the auxiliary variables satisfy the following conditions
\begin{align*}
\displaystyle\sum\limits_{x \in X_t} x = \displaystyle\sum\limits_{y \in Y_t} y &= 0 \\
\displaystyle\sum\limits_{x \in X_t} x^k - \displaystyle\sum\limits_{y \in Y_t} y^k  &= b_t^k - a_t^k \text{ for every } k \in \{2,\ldots,d\}.
\end{align*}
\end{lemma}

\begin{lemma} 
\label{lemma:soundness-structural}
For any subset $S \subseteq  \bigcup\limits_{t\in[n]} X_t \cup Y_t$ of the auxiliary variables, either
\[ \lvert \displaystyle\sum\limits_{y \in S} y \rvert > 10^{m+2n+\nu}~~ \text{ or } ~~\lvert \displaystyle\sum\limits_{y \in S} y \rvert < 10^{\nu}. \] 
\end{lemma}

We restate the following lemma from \Cref{sec:reduction}.

\LemRuntime*


In order to prove 
Lemma~\ref{lemma:condition12}, 
Lemma~\ref{lemma:soundness-structural} and 
Lemma~\ref{lemma:runtime} 
we first state some properties of the auxiliary variables generated by the \atomicsolverproc{t}{i} and prove them in 
Section~\ref{sec:helpers}. 

\begin{proposition} \label{prop:condition34}
For any $(t,i) \in [n] \times \{2, \dots, d\}$,  \atomicsolverproc{t}{i} on input  
a rational $R_{t,i}$, returns two sets  of auxiliary variables $\{ x_{t,i,j} \mid j \in [2^{i-1}] \}$ and $\{ y_{t,i,j} \mid j \in [2^{i-1}] \}$ which satisfy: 
\begin{align*}
\displaystyle\sum\limits_{j=1}^{2^{i-1}}  (x^i_{t,i,j} - y^i_{t,i,j})  &= R_{t, i}, \\ \displaystyle\sum\limits_{j=1}^{2^{i-1}}  (x_{t,i,j}^k - y_{t,i,j}^k) &= 0 \text{ for every } k \in \{1,\ldots, i-1\}.
\end{align*}
\end{proposition}

\begin{proposition}
\label{prop:sum-xtij-0}
For any $t \in [n]$, and $i \in \{2, 3, \dots, d\} $, 
\[\displaystyle\sum\limits_{j=1}^{2^{i-1}} x_{t,i,j} = \displaystyle\sum\limits_{j=1}^{2^{i-1}} y_{t,i,j} = 0. \]
\end{proposition}

\begin{proposition}
\label{prop:bound-on-alphas}
For any $t \in [n], i \in \{2, \dots, d\}$, we have
\begin{enumerate}[(a)]
\item \label{prop:alpha-condition} $i! \prod\limits_{r=1}^i \alpha_{t, i, r} = R_{t,i}$ 
\item \label{prop:bound-on-alpha-ti1} $ 10^{n^4} < \alpha_{t,i,1} < 10^{d! n^5} $
\item \label{prop:bound-on-alpha-tir} $ \alpha_{t,i,r} < 10^{nd^2}$  for $1 < r < i-1$
\item \label{prop:bound-on-alpha-tii1}$ \lvert \alpha_{t,i,i} \rvert < 2 $ 
\item \label{prop:bound-on-alpha-tii} $ \displaystyle\sum\limits_{r=2}^{i} \lvert \alpha_{t,i,r} \rvert < 10^{\nu-nd}$.
\end{enumerate}
\end{proposition}

\begin{proposition}
\label{prop:bound-on-xtij}
For any $t \in [n]$, $i \in \{2, \dots, d\}$ and $j \in [2^{i-1}]$, we have that
\[ 10^{(i-1)! \nu_t} - 10^{\nu-nd} \leq 2 \cdot\lvert x_{t,i,j} \rvert  \leq 10^{(i-1)! \nu_t} + 10^{\nu-nd}. \]
The analogous statement also holds for $y_{t,i,j}$.
\end{proposition}

\begin{proposition}\label{prop:all-distinct}
We have that:
\begin{enumerate}
\item For every $(t_1,i_1,j_1) \neq (t_2,i_2,j_2)$,  we have that $x_{t_1, i_1, j_1} \neq x_{t_2, i_2, j_2}$.
\item For every $(t_1,i_1,j_1) \neq (t_2,i_2,j_2)$,  we have that $y_{t_1, i_1, j_1} \neq y_{t_2, i_2, j_2}$.
\item For every $(t_1,i_1,j_1), (t_2,i_2,j_2)$,  we have that $x_{t_1, i_1, j_1} \neq y_{t_2, i_2, j_2}$.
\end{enumerate}

\end{proposition} 

\subsection{Proof of Lemma ~\ref{lemma:condition12}}
\label{subsec:lemma:condition12}
We now prove Lemma~\ref{lemma:condition12} which implies Properties~\ref{p1}  and \ref{p2} of the auxiliary variables.
\begin{proof}[Proof of Lemma~\ref{lemma:condition12}]
From Proposition~\ref{prop:sum-xtij-0}, we have that for any 
$t \in [n]$, and $i \in \{2, 3, \dots, d\} $, 
$\displaystyle\sum\limits_{j=1}^{2^{i-1}} x_{t,i,j} = \displaystyle\sum\limits_{j=1}^{2^{i-1}} y_{t,i,j} = 0$. 
Summing the variables over all $i \in \{2, 3, \dots, d\} $, we get 
$$\displaystyle\sum\limits_{x \in X_t} x = \displaystyle\sum\limits_{y \in Y_t} y = 0.$$ 

For the second part of the lemma, 
 for any $k \in  \{2, \dots, d\}$
\begin{align*}
\displaystyle\sum\limits_{x \in X_t} x^k - \displaystyle\sum\limits_{y \in Y_t} y^k & = 
\displaystyle\sum\limits_{i=2}^d\sum\limits_{j=1}^{2^{i-1}}  (x_{t,i,j}^k - y_{t,i,j}^k) \\
& =  \displaystyle\sum\limits_{i=2}^{k-1} \sum\limits_{j=1}^{2^{i-1}}  (x_{t,i,j}^k - y_{t,i,j}^k) + \sum\limits_{j=1}^{2^{k-1}}  (x_{t,k,j}^k - y_{t,k,j}^k) + \displaystyle\sum\limits_{i=k+1}^d \sum\limits_{j=1}^{2^{i-1}}  (x_{t,i,j}^k - y_{t,i,j}^k)
\end{align*}
From the definition of the residual, $R_{t,k}$, the first term, $ \displaystyle\sum\limits_{i=2}^{k-1} \sum\limits_{j=1}^{2^{i-1}}  (x_{t,i,j}^k - y_{t,i,j}^k) = b_t^k - a_t^k - R_{t,k} $. 
Also, from Proposition~\ref{prop:condition34} it follows that $\displaystyle\sum\limits_{j=1}^{2^{k-1}}  (x_{t,k,j}^k - y_{t,k,j}^k)  = R_{t,k}$ and $\displaystyle\sum\limits_{i=k+1}^d \sum\limits_{j=1}^{2^{i-1}}  (x_{t,i,j}^k - y_{t,i,j}^k) = 0$. 
Substituting these values in the above equation, we get,
\[ 
\displaystyle\sum\limits_{x \in X_t} x^k - \displaystyle\sum\limits_{y \in Y_t} y^k = b_t^k - a_t^k
\]
\end{proof}


\subsection{Proof of Lemma ~\ref{lemma:soundness-structural}}
\label{subsec:lemma:soundness-structural}
Before we prove Lemma~\ref{lemma:soundness-structural}, we note that each auxiliary variable, $x_{t,i,j}$ and $y_{t,i,j}$ is a $(\pm \frac12)$-linear combination of the $\alpha_{t,i,r}$ variables. From Proposition~\ref{prop:bound-on-alphas} (\ref{prop:bound-on-alpha-ti1}), (\ref{prop:bound-on-alpha-tir}) we note that each variable $\alpha_{t,i,r}$ is either of small magnitude, i.e. $\lvert \alpha_{t,i,r} \rvert < 10^{nd^2}$ 
or of fairly large magnitude, i.e. $\lvert \alpha_{t,i,r} \rvert >10^{n^4}$. Also, we note that there is only one large magnitude term, i.e., $\alpha_{t,i,1}$, for every pair $(t, i) \in [n] \times \{2, \cdots, d \}$. 

Recall that  $D$ is the set of all the auxiliary variables \[D = \{ x_{t,i,j}, y_{t,i,j} \mid t \in [n], i \in \{2, 3, \dots, d\}, j \in [2^{i-1}] \}. \]
For any auxiliary variable $z \in D$, we can split $z$  into terms of the form $\pm \frac12 \alpha_{t,i,r}$ 
with large magnitudes and terms with small magnitudes. 
$$ z = z_U + z_L, $$
where $z_U$ is the 
term with large magnitude and $z_L$ is the linear combinations of terms with small magnitudes. 
We now state and prove two properties of the small magnitude sum and the large magnitude sum which will imply the proof of Lemma~\ref{lemma:soundness-structural}.

\begin{claim}\label{claim:small-magnitudes}
For any subset $S \subseteq D$,  
$\displaystyle\sum\limits_{z \in S} z_L < 10^{\nu}$.
\end{claim}
\begin{proof}
For any subset $S \subseteq D$, 
\[ \displaystyle\sum\limits_{z \in S} z_L \leq \frac12\displaystyle\sum\limits_{t=1}^n\displaystyle\sum\limits_{i=2}^d \displaystyle\sum\limits_{r=2}^i \lvert \alpha_{t,i,r} \rvert. \]
From Proposition~\ref{prop:bound-on-alphas}(\ref{prop:bound-on-alpha-tii}), we know that for any $(t, i) \in [n] \times \{2, \dots, d\}$, $\displaystyle\sum\limits_{r=2}^i \lvert \alpha_{t,i,r} \rvert \leq 10^{\nu - nd}$. 
Summing over all $(t,i)$, we upper bound the sum of small magnitude terms as follows:
\begin{align*}
\displaystyle\sum\limits_{z \in S} z_L & \leq \frac12\displaystyle\sum\limits_{t=1}^n\displaystyle\sum\limits_{i=2}^d \displaystyle\sum\limits_{r=2}^i \lvert \alpha_{t,i,r} \rvert \\
& \leq nd \cdot 10^{\nu - nd} \\
& < 10^\nu \qedhere
\end{align*}
\end{proof}

\begin{claim}\label{claim:large-magnitudes}
Let $S \subseteq D$ such that $\displaystyle\sum\limits_{z \in S} z_U \neq 0$, then 
$\left| \displaystyle\sum\limits_{z \in S} z_U \right| \geq  \frac12 \cdot 10^{n^4}$.
\end{claim}

\begin{proof}
We show that for any subset of the auxiliary variables, the contribution of the large magnitudes is either 0, or larger than $\frac12 \cdot 10^{n^4}$ .  
Note that all the large magnitude terms, i.e., $\alpha_{t, i, 1}$ for any $(t, i)$, are powers of $10$ larger than $n^4$ and therefore, each $z_U$, being a $\pm\frac12$ multiple of the large term, is divisible by $\frac12 \cdot 10^{n^4}$. Thus, the sum $\left| \displaystyle\sum\limits_{z \in S} z_U \right|$ is divisible by  $\frac12 \cdot 10^{n^4}$. If the sum is non-zero, then it is a non-zero multiple of $\frac12 \cdot 10^{n^4}$ and hence is larger than $\frac12 \cdot 10^{n^4}$.

\end{proof}

\noindent The proof of Lemma~\ref{lemma:soundness-structural} now follows by combining Claim~\ref{claim:small-magnitudes} and Claim~\ref{claim:large-magnitudes}. 
\begin{proof}[Proof of Lemma~\ref{lemma:soundness-structural}]
For any subset $S\subseteq D$,  we can split the sum of the variables as:
\[\displaystyle\sum\limits_{z \in S} z =  \displaystyle\sum\limits_{z \in S} z_U +  \displaystyle\sum\limits_{z \in S} z_L. \]

If $\displaystyle\sum\limits_{z \in S} z_U \neq 0$, then from Claim~\ref{claim:small-magnitudes} and Claim~\ref{claim:large-magnitudes} we have,
\begin{align*}
 \left| \displaystyle\sum\limits_{z \in S } z \right| & \geq   \left| \displaystyle\sum\limits_{z \in S} z_U  \right| -   \left| \displaystyle\sum\limits_{z \in S } z_L  \right| \\
 & \geq \frac12 \cdot 10^{n^4} - 10^{\nu} = \Omega(10^{n^4})> 10^{m+2n +\nu} ~~\text{ [using the fact that $\nu = n^2$]}.
 \end{align*}
 On the other hand, if $\displaystyle\sum\limits_{z \in S} z_U  = 0$, then from Claim~\ref{claim:small-magnitudes}, 
 \[\left|\displaystyle\sum\limits_{z \in S} z \right| = \left| \displaystyle\sum\limits_{z \in S } z_L  \right|
 \leq 10^{\nu}. \]
\end{proof}


\subsection{Proof of Lemma ~\ref{lemma:runtime}}
\label{subsec:lemma:runtime}

\begin{proof} [Proof of Lemma~\ref{lemma:runtime}]
In the construction of the instance of \mss{d}, we create $2$ variables, i.e., $a_t, b_t$ and $2^{d+1}-4$ auxiliary variables $X_t \cup Y_t$ corresponding to each of the $n$ literals in the 1-in-3 SAT instance.  From Claim~\ref{claim:all-distinct} below, we know that all variables in the set $A$ are distinct. Therefore, the size of the set $A$ in the instance of \mss{d}, $N = n(2^{d+1}-2)$.  Now we show that every element constructed in the instance of \mss{d}  
has $\poly(n, d!)$ digit representation. 

From Fact~\ref{claim:bound-at}, Proposition~\ref{prop:bound-on-xtij} and Claim~\ref{claim:bound-Bk} below, we know that the magnitudes of all the numbers generated by the reduction are bounded by $10^{\poly(n,d!)}$. Therefore to complete the proof, it remains to show that the denominators of all the rational numbers in the instance of \mss{d} are also bounded by $10^{\poly(n,d!)}$

Observe from Definition~\ref{eqn:def-vars} that  $a_t$ and $b_t$ for every $t \in [n]$ are integers. Also, for any $t \in [n]$ and $i \in \{2, \cdots, d\}$ each $\alpha_{t, i, r}$ for $ 1 \leq r \leq i-1$ constructed by \atomicsolverproc{t}{i} is a unique power of $10$, and hence an integer, but  $\alpha_{t,i,i}$ is a rational number. Each auxiliary variable generated by \atomicsolverproc{t}{i} is therefore a rational number due to the contribution from $\alpha_{t, i, i}$. 
From Claim~\ref{bounddr} below, it follows that every rational number in the instance of \mss{d} has magnitude at most $10^{\poly(n, d!)}$ and therefore a $\poly(n, d!)$ digit representation. 
\end{proof}

The following claim bounds the magnitudes of the targets in the \mss{d} instance.   
\begin{claim}
\label{claim:bound-Bk}
For every $k \in \{ 2, \cdots, d \}$, 
\[ \lvert B_k \rvert \leq 10^{k(d!)n^6}. \]
\end{claim}
\begin{proof}
Recall from Definition~\ref{eqn:def-targets}, 
\[B_k = \displaystyle\sum\limits_{t = 1}^n a_t^k + \displaystyle\sum\limits_{t = 1}^n\displaystyle\sum\limits_{x \in X_t} x^k \text{   for every } k \in \{2, \dots, d \}.
\]
Using bounds on the magnitudes of $a_t$ and $x \in X_t$ from Fact~\ref{claim:bound-at} and Proposition~\ref{prop:bound-on-xtij} 
we get 
\begin{align*}
\lvert B_k \rvert &\leq n(10^{k(m+n+\nu+1)}) + n2^d((10^{ (d-1)!\nu_n} + 10^{\nu-nd} )^k) \\
&\leq 10^{k(d!)n^6}. \qedhere
\end{align*}

\end{proof}

We now bound the magnitude of the denominators of $\alpha_{t, i, i}$ for every $(t, i) \in [n] \times \{2, \cdots, d\}$. This bound will be used in Claim~\ref{bounddr} to bound the denominators of all the rational numbers in the instance of \mss{d}. Let $D(x)$ denote the irreducible denominator of a rational number $x$. 
\begin{claim}
\label{alpha-tii-dr}
For any $(t, i) \in [n] \times \{2, \cdots, d\}$, 
\[D(\alpha_{t, i, i}) \leq 10^{(i!)^2\cdot n^{6}} \]
\end{claim}
\begin{proof}
The proof proceeds by first obtaining a recursive expression for $D(\alpha_{t, i,i})$, and we then use induction on $i$ to show the bound. 
Recall the definition of $\alpha_{t, i, i}$ from Algorithm~\ref{algo:gadget},
\[ \alpha_{t,i,i} = \frac{R_{t,i}}{ i ! \prod\limits_{ r\in [i-1]} \alpha_{t,i,r} }, \]
where $R_{t,i}$ is defined as 
\[ R_{t,i} = b_t^i - a_t^i + \displaystyle\sum\limits_{u=2}^{i-1}\displaystyle\sum\limits_{v=1}^{2^{u-1}} (y_{t,u,v}^i - x_{t,u,v}^i). \]
Therefore, it follows that the denominator of $\alpha_{t, i, i}$ is bounded by the product of the denominator of $R_{t,i}$ and $i! \cdot \prod\limits_{r =1}^{i-1} \alpha_{t,i,r}$. i.e.,
\begin{align*}
D(\alpha_{t, i, i}) &\leq D(R_{t, i})\cdot(i! \cdot \prod\limits_{r =1}^{i-1} \alpha_{t,i,r} ) \\
&= D(R_{t, i})\cdot(i! \cdot 10^{(i-1)! \nu_t + \sum\limits_{r = 2}^{i-1} g(t,i,r)} )\\
&\leq D(R_{t, i})\cdot(i! \cdot 10^{(i-1)! n^5 + nd^3})
\end{align*}
The last inequality follows from the fact that $\sum\limits_{r = 2}^{i-1} g(t,i,r)  = \sum\limits_{r = 2}^{i-1} (t-1)d^2+(i-1)i + r \leq t d^3$ for all $ 2 \leq i \leq d$ and $\nu_t < n^5$ for any $t \in [n]$. We now obtain an expression for $D(R_{t,i})$. Since $b_t$ and $a_t$ are both integers, note that 
$D(R_{t,i}) = D\left(\displaystyle\sum\limits_{u=2}^{i-1}\displaystyle\sum\limits_{v=1}^{2^{u-1}} (y_{t,u,v}^i - x_{t,u,v}^i)\right)$. 
Also, recall that all the auxiliary variables obtained from a given \atomicsolverproc{t}{u}, described in Algorithm~\ref{algo:gadget}, have the same denominator to which $D(\alpha_{t, u, u})$ contributes, i.e.,  
$D(x_{t, u, v}) = D(y_{t, u, v}) = 2 \cdot D(\alpha_{t, u, u})$, for all $v \in [2^{u-1}]$.
Therefore,  $D(y_{t,u,v}^i - x_{t,u,v}^i ) = 2^i \cdot D( \alpha_{t, u, u}^i)$. 
From this observation, it follows that
$ D( \displaystyle\sum\limits_{v=1}^{2^{u-1}} y_{t,u,v}^i - x_{t,u,v}^i)  = 2^i \cdot D( \alpha_{t, u, u}^i)$, 
and we get an expression for $D(R_{t,i})$ as follows: 
\[ D(R_{t,i}) = LCM( \{ 2^i \cdot D(  \alpha_{t, u, u}^i)  \mid u \in \{2, \cdots, i-1\} \} ) \leq 2^{i^2} \cdot \displaystyle\prod_{u =2}^{i-1} D( \alpha_{t, u, u}^i).  \]
Substituting the above expression for $D(R_{t,i})$ back in the expression obtained for $D(\alpha_{t,i,i})$, we get  
\begin{equation}
\label{denom-eq}
 D(\alpha_{t, i, i}) \leq \left( \displaystyle\prod_{u =2}^{i-1} D( \alpha_{t, u, u}^i) \right) \cdot(2^{i^2} \cdot i! \cdot 10^{(i-1)! n^5 + nd^3})
 \end{equation}
 
We now use induction on $i$ to show that that $D(\alpha_{t, i, i}) \leq 10^{(i!)^2\cdot n^{6}}$ for every $i \in \{2, \cdots, d\}$. For the base case, $i=2$, from definitions we know that
\[D(\alpha_{t, 2, 2}) = 2 \cdot10^{\nu_t} < 10^{n^6}\]

Let us assume the induction hypothesis that for all $i < \ell \leq d$,
\[ D(\alpha_{t, i, i}) \leq 10^{(i !)^2 \cdot n^{6}}. \]

From Equation~\ref{denom-eq}, we know that
\begin{align*}
 D(\alpha_{t, \ell, \ell}) &\leq \left( \displaystyle\prod_{u=2}^{\ell-1} D( \alpha_{t, u, u}^\ell)\right) \cdot (2^{\ell^2} \cdot \ell! \cdot 10^{(\ell-1)! n^5 + nd^3}) \\
 &\leq \left( \displaystyle\prod_{u=2}^{\ell-1} (10^{(u !)^2 \cdot n^{6}} )^\ell\right)\cdot (10^{(\ell-1)! n^5 + nd^3 + 2 \ell^2}) \\
 &\leq 10^{ \ell \sum\limits_{ u=2}^{\ell-1} ( (u !)^2 \cdot n^{6} ) + (\ell)! n^5 + nd^3+ 2\ell^2} \\
  &\leq 10^{ \ell \cdot (\ell-1) \cdot (\ell-1)!^2 \cdot n^{6}  + (\ell)! n^5 +nd^3 + 2 \ell^2} \\
  & \leq 10^{ (\ell !)^2 \cdot n^{6}},
 \end{align*}
where the last inequality follows from the fact that $ \ell (\ell-1 )!^2 n^6 > (\ell)! n^5 + nd^3 + 2\ell^2 $ for any $\ell \leq d$.
\end{proof}

\begin{claim}
\label{bounddr}
For any $x \in  A  ~\bigcup~ \{ B_1, \cdots, B_d \} $, 
\[  D(x) < 10^{\poly(n, d!)}. \]
\end{claim}

\begin{proof} [Proof of Claim~\ref{bounddr}]
We first observe that the elements constructed from the 3-SAT clauses and variables are all integers. So, $D(a_t) = D(b_t) = 1$ for all $t \in [n]$. Next, we argue about the denominators of the auxiliary variables and show that they are all bounded by $2 \cdot 10^{(d!)^2\cdot n^{6}}$. Consider the set of auxiliary variables generated by \atomicsolverproc{t}{i} for some $t \in [n]$ and $i \in \{2, 3, \cdots, d\}$. Each $x_{t,i,j}$ (or $y_{t,i,j}$) 
 is a $\pm \frac12$-linear combination of the $\{ \alpha_{t,i,r} \mid r \in [i] \}$ variables. From the definitions in Algorithm~\ref{algo:gadget}, we note that all $\alpha_{t, i, r}$ variables constructed by the \atomicsolver are integers except for $\alpha_{t,i,i}$. Therefore, each $x_{t,i,j}$ and $y_{t,i,j}$ have the same denominator as $\alpha_{t,i,i}/2$. Using Claim~\ref{alpha-tii-dr}, we get that for every $(t,i)$, $D(\alpha_{t, i, i,}) \leq 10^{(i!)^2\cdot n^{6}}$. Therefore, for any $j \in [2^{i-1}]$,  $D(x_{t, i, j}) < 2 \cdot 10^{(i!)^2\cdot n^{6}}$. A similar argument applies to $y_{t,i,j}$. 

We now bound the magnitudes of the denominators of the target, $B_1, \cdots, B_d$ defined in the \mss{d} instance. Recall from Definition~\ref{eqn:def-targets} that $B_1$ is an integer. Therefore, $D(B_1) = 1$. All other targets are rational numbers defined as
\[B_k = \displaystyle\sum\limits_{t = 1}^n a_t^k + \displaystyle\sum\limits_{t = 1}^n\displaystyle\sum\limits_{x \in X_t} x^k \text{   for every } k \in \{2, \dots, d \}.
\]
The denominator of $B_k$ is defined by the denominator of the sum, $\displaystyle\sum\limits_{t = 1}^n\displaystyle\sum\limits_{x \in X_t} x^k$. This sum can be expanded as $\displaystyle\sum\limits_{t= 1}^n\displaystyle\sum\limits_{i= 2}^d\displaystyle\sum\limits_{j= 1}^{2^{i-1}}  x^k$.  From the fact that $D(\displaystyle\sum\limits_{j= 1}^{2^{i-1}}  x^k) = D(\alpha_{t,i,i}^k)$ and Claim~\ref{alpha-tii-dr}, we get, 
\begin{align*}
D(B_k) &\leq \displaystyle\prod\limits_{\substack{ t \in [n]\\  i \in \{2 \ldots, d\}}} D(\alpha_{t,i,i}^k) \\
&\leq \displaystyle\prod\limits_{\substack{ t \in [n]\\  i \in \{2 \ldots, d\}}} 10^{k (i!)^2\cdot n^{6}}\\
&\leq (10^{k (d!)^2\cdot n^{6}})^{nd} = 10^{kd (d!)^2\cdot n^{7}}
\end{align*}

Therefore, we conclude that every element of the instance of \mss{d} constructed by the reduction has a denominator of magnitude at most $10^{\poly(n, d!)}$. 
\end{proof}

\begin{claim}
\label{claim:all-distinct}
All variables in the set $A$ are distinct. 
\end{claim}
\begin{proof}[Proof of Claim~\ref{claim:all-distinct}]
From Proposition~\ref{prop:all-distinct}, we know that all auxiliary variables are distinct. Also, the distinctness of the variables $\{a_t, b_t\mid t \in [n] \}$ follows from the construction. The only thing that remains to show is that all the auxiliary variables are different from $\{a_t, b_t\mid t \in [n] \}$.

We show this fact by comparing the magnitudes of the two sets of variables. From Fact~\ref{claim:bound-at}, we know that 
$\lvert v \rvert < 10^{m+n+\nu+1} $ for every $v \in \{a_t, b_t\mid t \in [n] \}$, and from Proposition~\ref{prop:bound-on-xtij}, we know that all auxiliary variables are larger than $10^{\nu_1} - 10^{\nu-nd} > 10^{m+n+\nu+1}$. Therefore the two sets of variables are disjoint. 
\end{proof}

\section{Proofs of the Helper Propositions~\ref{prop:condition34}, \ref{prop:sum-xtij-0}, \ref{prop:bound-on-alphas}, \ref{prop:bound-on-xtij}, \ref{prop:all-distinct} }
\label{sec:helpers}

In this section, we prove the helper claims stated in the previous section. 

\begin{proof}[Proof of Proposition~\ref{prop:condition34}]
We first show a structural property of the auxiliary variables generated by any $\atomicsolver$. The proof of Proposition~\ref{prop:condition34} follows from it. 

\begin{claim}\label{claim:structural-p1}
For any $i \in \{2, \cdots, d\}$,
Let $A_i$, $B_i$ are matrices defined in the \atomicsolver and let $\{ \alpha_{r} \mid r \in [i] \}$ be some rational numbers.
If 
\begin{align*}
\begin{bmatrix} x_{1} \\ x_{2}\\ \vdots \\ x_{2^{i-1}} \end{bmatrix} = \frac12 \cdot A_i \cdot \begin{bmatrix} \alpha_{1} \\  \alpha_{2} \\ \vdots \\ \alpha_{i} \end{bmatrix}, 
\text{ and }
\begin{bmatrix} y_{1} \\ y_{2}\\ \vdots \\ y_{2^{i-1}} \end{bmatrix} = \frac12 \cdot B_i \cdot \begin{bmatrix} \alpha_{1} \\  \alpha_{2} \\ \vdots \\ \alpha_{i}\end{bmatrix},
\end{align*}

then $\{ x_{j} \mid j \in [2^{i-1}] \}$ and $\{ y_{j} \mid j \in [2^{i-1}] \}$  satisfy: 
\begin{align*}
\displaystyle\sum\limits_{j=1}^{2^{i-1}}  (x_{j}^k - y_{j}^k) &= 0 \text{ for every } k \in \{1,\ldots, i-1\} \\
\displaystyle\sum\limits_{j=1}^{2^{i-1}}  (x^i_{j} - y^i_{j})  &= i! \prod\limits_{r=1}^i \alpha_{r} \\
\end{align*}
\end{claim}

\begin{proof}[Proof of Claim~\ref{claim:structural-p1}]
We use induction on $i$. 
For the base case, consider $i=2$. From the definition of $A_2$ and $B_2$ we get,
\begin{align*}
x_{1} &= \frac{\alpha_{1}}{2} + \frac{\alpha_{2}}{2} \\
x_{2} &= -\frac{\alpha_{1}}{2} - \frac{\alpha_{2}}{2} \\
y_{1} &= \frac{\alpha_{1}}{2} - \frac{\alpha_{2}}{2} \\
y_{2} &= -\frac{\alpha_{1}}{2} + \frac{\alpha_{2}}{2}  
\end{align*}
Therefore, 
\begin{align*}
x_{1}+x_{2}-y_{1}-y_{2} &= 0 \\
x_{1}^2+x_{2}^2-y_{1}^2-y_{2}^2 &= 2 \cdot \alpha_{1}\cdot \alpha_{2} 
\end{align*}
and the claim holds for $i = 2$. 

Let us assume the induction hypothesis for all $i < \ell \leq d$. 
For $i = \ell$, we have,
\begin{align*}
 \begin{bmatrix} x_{1} \\ x_{ 2}\\ \vdots \\ x_{2^{\ell-1}} \end{bmatrix} = \frac12 \cdot A_{\ell} \cdot \begin{bmatrix} \alpha_{ 1} \\  \alpha_{2} \\ \vdots \\ \alpha_{ \ell} \end{bmatrix}
 \text{ and  }
\begin{bmatrix} y_{1} \\ y_{ 2}\\ \vdots \\ y_{2^{\ell-1}} \end{bmatrix} = \frac12 \cdot B_{\ell} \cdot \begin{bmatrix} \alpha_{ 1} \\  \alpha_{2} \\ \vdots \\ \alpha_{ \ell}\end{bmatrix}
\end{align*}

From the recursive definitions of the matrices $A_{\ell}$, $B_{\ell}$ in Algorithm~\ref{algo:gadget}, we can split the above equations as
\begin{align*}
 \begin{bmatrix} x_{1} \\ x_{ 2}\\ \vdots \\ x_{2^{\ell-2}} \end{bmatrix} &= \frac12 \cdot A_{\ell -1} \cdot \begin{bmatrix} \alpha_{1} \\  \alpha_{2} \\ \vdots \\ \alpha_{\ell-1} \end{bmatrix} + \frac12 \cdot \begin{bmatrix} \alpha_{\ell} \\  \alpha_{\ell} \\ \vdots \\ \alpha_{ \ell} \end{bmatrix}  \\
 \begin{bmatrix} x_{2^{\ell-2}+1} \\ x_{ 2^{\ell-2}+2}\\ \vdots \\ x_{2^{\ell-1}} \end{bmatrix} &= \frac12 \cdot B_{\ell -1} \cdot \begin{bmatrix} \alpha_{  1} \\  \alpha_{  2} \\ \vdots \\ \alpha_{  \ell-1} \end{bmatrix} -  \frac12 \cdot \begin{bmatrix} \alpha_{  \ell} \\  \alpha_{  \ell} \\ \vdots \\ \alpha_{\ell} \end{bmatrix}  \\
 \begin{bmatrix} y_{1} \\ y_{ 2}\\ \vdots \\ y_{2^{\ell-2}} \end{bmatrix} &= \frac12 \cdot B_{\ell -1} \cdot \begin{bmatrix} \alpha_{ 1} \\  \alpha_{  2} \\ \vdots \\ \alpha_{\ell-1} \end{bmatrix} + \frac12 \cdot \begin{bmatrix} \alpha_{ \ell} \\  \alpha_{ \ell} \\ \vdots \\ \alpha_{ \ell} \end{bmatrix}  \\
 \begin{bmatrix} y_{2^{\ell-2}+1} \\ y_{2^{\ell-2}+2}\\ \vdots \\ y_{2^{\ell-1}} \end{bmatrix} &= \frac12 \cdot A_{\ell -1} \cdot \begin{bmatrix} \alpha_{ 1} \\  \alpha_{ 2} \\ \vdots \\ \alpha_{\ell-1} \end{bmatrix} - \frac12 \cdot  \begin{bmatrix} \alpha_{ \ell} \\  \alpha_{ \ell} \\ \vdots \\ \alpha_{ \ell} \end{bmatrix}  \\
\end{align*}

Equivalently, they can be rewritten as
\[ x_{ j} = \left\{ \begin{array}{ll}
		  &  x_{ j}' + \frac12 \cdot  \alpha_{  \ell} ~ \mbox{if }  j \leq 2^{\ell-2} \\
		  &  y_{ j-2^{\ell-2}}' - \frac12 \cdot \alpha_{  \ell} ~ \mbox{ if } j > 2^{\ell-2} \\
	\end{array} \right. \]
Similarly, 
\[ y_{j} = \left\{ \begin{array}{ll}
		  &  y_{  j}' +\frac12 \cdot  \alpha_{\ell} ~ \mbox{if }  j \leq 2^{\ell-2} \\
		  &  x_{j-2^{\ell-2}}' -\frac12 \cdot  \alpha_{  \ell} ~ \mbox{ if } j > 2^{\ell-2} \\
	\end{array} \right. \]
where, the $\{ x_{ j}',  y_{  j}' \mid j \in [2^{\ell-2}]\} $ by induction hypothesis satisfy
\begin{align*}
\displaystyle\sum\limits_{j=1}^{2^{\ell-2}}  (x^{'\ell-1}_{j} - y^{'\ell-1}_{j})  &= (\ell-1)!\prod\limits_{r=1}^{\ell-1} \alpha_{ r} \\
\displaystyle\sum\limits_{j=1}^{2^{\ell-2}}  (x_{j}^{'k} - y_{j}^{'k}) &= 0 \text{ for every } k \in \{1,\ldots, \ell-2\}.
\end{align*}

Therefore, for any $k \in \mathbb{N}$, we get that 
\begin{align*}
\displaystyle\sum\limits_{j=1}^{2^{\ell-1}}  (x_{j}^k - y_{j}^k) = 
& \displaystyle\sum\limits_{j=1}^{2^{\ell-2}}  (x_{j}' + \frac12 \cdot \alpha_{\ell})^k - (y_{j}' + \frac12 \cdot \alpha_{\ell})^k \\
&+ \displaystyle\sum\limits_{j= 2^{\ell-2}+1}^{2^{\ell-1}} (y_{j-2^{\ell-2}}' - \frac12 \cdot \alpha_{ \ell})^k - (x_{j-2^{\ell-2}}' - \frac12 \cdot \alpha_{\ell})^k \\
= &\displaystyle\sum\limits_{j=1}^{2^{\ell-2}}  (x_{j}' + \frac12 \cdot \alpha_{\ell})^k - (x_{j}' - \frac12 \cdot \alpha_{\ell})^k \\
& - \displaystyle\sum\limits_{j=1}^{2^{\ell-2}} (y_{j}' + \frac12 \cdot \alpha_{\ell})^k - (y_{j}' - \frac12 \cdot \alpha_{ \ell})^k \\
= &\displaystyle\sum\limits_{j=1}^{2^{\ell-2}} \left( 2 \displaystyle\sum_{\substack{r=0 \\ r \equiv 1 \mmod 2}}^k  \frac{1}{2^r} \cdot \dbinom{k}{r} x_{j}^{'k-r} \alpha_{\ell}^{r} \right) \\
&- \displaystyle\sum\limits_{j=1}^{2^{\ell-2}} \left( 2 \displaystyle\sum_{\substack{r=0 \\ r \equiv 1 \mmod 2}}^k \frac{1}{2^r} \cdot \dbinom{k}{r} y_{j}^{'k-r} \alpha_{\ell}^{r} \right) \\
= & \displaystyle\sum\limits_{j=1}^{2^{\ell-2}}\left( 2 \displaystyle\sum_{\substack{r=0 \\ r \equiv 1 \mmod 2}}^k \frac{1}{2^r} \cdot \dbinom{k}{r} ( x_{ j}^{'k-r} - y_{ j}^{'k-r} )\alpha_{\ell}^{r} \right) 
\end{align*}
 
Observe that, for all $k \leq \ell-1$, $k-r \leq \ell-2$, since $r \equiv 1 \mmod 2$. Therefore, for all $k \leq \ell-1$, from induction hypothesis,  
we have, $( x_{ j}^{'k-r} - y_{ j}^{'k-r} ) = 0$. And, 
\[ \displaystyle\sum\limits_{j=1}^{2^{\ell-1}}  (x_{j}^k - y_{j}^k) = 0.
\] 
For $k = \ell$, 
\begin{align*}
\displaystyle\sum\limits_{j=1}^{2^{\ell-1}}  (x_{j}^\ell - y_{j}^\ell) = 
&\displaystyle\sum\limits_{j=1}^{2^{\ell-2}}\left( 2 \displaystyle\sum_{\substack{r=0 \\ r \equiv 1 \mmod 2}}^{\ell} \frac{1}{2^r} \cdot \dbinom{\ell}{r} ( x_{ j}^{'\ell-r} - y_{j}^{'\ell-r} )\alpha_{ \ell}^{r} \right)  \\
=& 2 \cdot  \frac{1}{2} \cdot\dbinom{\ell}{1}\cdot \alpha_{ \ell}   \displaystyle\sum\limits_{j=1}^{2^{\ell-2}} ( x_{j}^{'\ell-1} - y_{ j}^{'\ell-1} )\\
= &  \ell \cdot (\ell-1)! \cdot \prod\limits_{r=1}^{\ell-1} \alpha_{ r} \cdot \alpha_{\ell}\\
= & \ell! \cdot \prod\limits_{r=1}^{\ell} \alpha_{ r}.
\end{align*}
\end{proof}

Note that Claim~\ref{claim:structural-p1} is independent of $t$ and the choice of the $\alpha$ variables.  
Recall the construction of \atomicsolverproc{t}{i} for any $(t,i) \in [n] \times \{2, \cdots, d\}$. It returns two sets  of auxiliary variables $\{ x_{t,i,j} \mid j \in [2^{i-1}] \}$ and $\{ y_{t,i,j} \mid j \in [2^{i-1}] \}$ which are constructed using matrices $A_i$ and $B_i$. From Claim~\ref{claim:structural-p1}, it then follows that these auxiliary variables satisfy:
\begin{align*}
\displaystyle\sum\limits_{j=1}^{2^{i-1}}  (x^i_{tij} - y^i_{tij})  &= i! \prod\limits_{r=1}^i \alpha_{t, i, r} \\
\displaystyle\sum\limits_{j=1}^{2^{i-1}}  (x_{tij}^k - y_{tij}^k) &= 0 \text{ for every } k \in \{1,\ldots, i-1\}
\end{align*}

Using Proposition~\ref{prop:bound-on-alphas} (\ref{prop:alpha-condition}),  
we get,
\[
\displaystyle\sum\limits_{j=1}^{2^{i-1}}  (x^i_{tij} - y^i_{tij})  = b_t^i - a_t^i + R_{t,i}
\]
\end{proof}


\begin{proof}[Proof of Proposition~\ref{prop:sum-xtij-0}]
The proof uses the recursive structure of the matrices $A_i$ and $B_i$. Recall that $\bf{1}^\ell$ denotes a vector of $\ell$ ones, and ${\bf 0}^{\ell}$ denotes a vector of $\ell$ zeros. Note that for any $(t, i) \in [n] \times \{2, \cdots, d\}$, 

\[ \displaystyle\sum\limits_{j = 1}^{2^{i-1}} x_{t,i,j} = \frac12 \cdot 
({\bf 1}^{2^{i-1}})^T \cdot A_i \cdot \begin{bmatrix} \alpha_{t,i,1} & \cdots & \alpha_{t,i,i} \end{bmatrix}^T.\]
Similarly, the sum of all the $\{ y_{t, i,j} \mid j \in [2^{i-1}] \}$ can be written as
\[ \displaystyle\sum\limits_{j = 1}^{2^{i-1}} y_{t,i,j} = \frac12 \cdot 
({\bf 1}^{2^{i-1}})^T \cdot B_i \cdot \begin{bmatrix} \alpha_{t,i,1} & \cdots & \alpha_{t,i,i} \end{bmatrix}^T.\]
We show by induction on $i \geq 2$ that 
\[({\bf 1}^{2^{i-1}})^T \cdot A_i = ({\bf 0}^{i })^T \text{ and } ({\bf 1}^{2^{i-1}})^T \cdot B_i = ({\bf 0}^{i})^T. \]

For the base case, $i=2$, it can be verified that 
\begin{align*}
\begin{bmatrix} 1 & 1 \end{bmatrix} \cdot A_2 =  \begin{bmatrix} 1 & 1 \end{bmatrix} \cdot  \begin{bmatrix} 1 & 1 \\ -1 & -1 \end{bmatrix} &= \begin{bmatrix} 0 & 0 \end{bmatrix} \text{ and }\\
 \begin{bmatrix} 1 & 1 \end{bmatrix} \cdot B_2 = \begin{bmatrix} 1 & 1 \end{bmatrix} \cdot  \begin{bmatrix} 1 & -1 \\ -1 & 1 \end{bmatrix} &= \begin{bmatrix} 0 & 0 \end{bmatrix}
 \end{align*}
 
Let us assume the induction hypothesis for all $i <\ell \leq d$. For $i = \ell$, observe that 
 \begin{align*}
 ({\bf 1}^{2^{\ell-1}})^T \cdot A_i &= 
 \begin{bmatrix} ({\bf 1}^{2^{\ell-2}})^T & ({\bf 1}^{2^{\ell-2}})^T \end{bmatrix} \cdot  \begin{bmatrix} A_{\ell-1} & {\bf 1}^{2^{\ell-2}}\\ B_{i-1} & -{\bf 1}^{2^{\ell-2}} \end{bmatrix} \\
 &= \begin{bmatrix} ({\bf 1}^{2^{\ell-2}})^T \cdot A_{\ell-1} + ({\bf 1}^{2^{\ell-2}})^T \cdot B_{\ell-1} & 0 \end{bmatrix} 
 \end{align*}
By the induction hypothesis, we know that $({\bf 1}^{2^{\ell-2}})^T \cdot A_{\ell-1} + ({\bf 1}^{2^{\ell-2}})^T \cdot B_{\ell-1}   = ({\bf 0}^{\ell -1})^T $, 
Therefore, 
\[  ({\bf 1}^{2^{\ell-1}})^T\cdot A_\ell = \begin{bmatrix} {\bf 0}^{\ell -1} &0 \end{bmatrix}
 \]
 Similarly,  
 \[  ({\bf 1}^{2^{\ell-1}})^T\cdot B_\ell = \begin{bmatrix} ({\bf 1}^{2^{\ell-2}})^T \cdot B_{\ell-1} + ({\bf 1}^{2^{\ell-2}})^T \cdot A_{\ell-1} & 0 \end{bmatrix}  = \begin{bmatrix} {\bf 0}^{\ell -1} &0 \end{bmatrix}
 \]

\end{proof}

We now show certain bounds on the magnitudes of $\alpha_{t,i,r}$ and hence on the auxiliary variables $x_{t,i,j}, y_{t,i,j}$. For any two tuples of same dimensions, we say that $(p_1, p_2, \cdots, p_d) > (q_1, q_2, \cdots, q_d)$ if there is an $i \in [d]$ such that $p_i > q_i$ and $p_j = q_j$ for all $j < i$. In order to prove \Cref{prop:bound-on-alphas}, we will need the following claim.

\begin{claim}\label{claim:diff-x-y}
For any $t \in [n]$, $i \in \{2, \dots, d\}$ and any $j \in [2^{i-1}]$, 
\[ \lvert x_{t,i,j} - y_{t,i,j} \rvert =  \alpha_{t,i, 2} \] 
\end{claim}
\begin{proof}[Proof of Claim~\ref{claim:diff-x-y}]
We use the recursive matrix definitions to show that for every $t, i, j \in [n] \times \{2, \cdots, d\} \times [2^{i-1}]$,
\[ \lvert x_{t,i,j} - y_{t,i,j} \rvert = \alpha_{t, i, 2}\]
We use induction on $i$. 

For the base case, $i = 2$, 
\begin{align*}
x_{t,2,1} - y_{t,2,1}  &= \frac{\alpha_{t,2,1}}{2} + \frac{\alpha_{t,2,2}}{2} -\frac{\alpha_{t,2,1}}{2} + \frac{\alpha_{t,2,2}}{2}  = \alpha_{t,2,2} \\
x_{t,2,2} - y_{t,2,2}  &= -\frac{\alpha_{t,2,1}}{2} - \frac{\alpha_{t,2,2}}{2}  + \frac{\alpha_{t,2,1}}{2} - \frac{\alpha_{t,2,2}}{2}  = - \alpha_{t,2,2} 
\end{align*}

Let us assume the induction hypothesis for all $i < \ell \leq d$. 

From the definition of the \atomicsolver we know that, 
\begin{align*}
\begin{bmatrix} x_{t, \ell,1} \\ x_{t, \ell, 2}\\ \vdots \\ x_{t, \ell, 2^{\ell-1}} \end{bmatrix} = \frac12 \cdot A_{\ell} \cdot \begin{bmatrix} \alpha_{t, \ell, 1} \\  \alpha_{t, \ell, 2} \\ \vdots \\ \alpha_{t, \ell, \ell} \end{bmatrix}, 
\text{ and }
\begin{bmatrix} y_{t, \ell,1} \\ y_{t, \ell, 2}\\ \vdots \\ y_{t, \ell, 2^{\ell-1}} \end{bmatrix} = \frac12 \cdot B_{\ell} \cdot \begin{bmatrix} \alpha_{t, \ell, 1} \\  \alpha_{t, \ell, 2} \\ \vdots \\ \alpha_{t, \ell, \ell}\end{bmatrix}. 
\end{align*}

Therefore, 
\[
\begin{bmatrix} x_{t, \ell,1}- y_{t, \ell, 1} \\ x_{t, \ell, 2} - y_{t, \ell, 2} \\ \vdots \\ x_{t, \ell, 2^{\ell-1}}-y_{t, \ell, 2^{\ell-1}} \end{bmatrix} = \frac12 \cdot (A_{\ell} - B_{\ell}) \cdot \begin{bmatrix} \alpha_{t, \ell, 1} \\  \alpha_{t, \ell, 2} \\ \vdots \\ \alpha_{t, \ell, \ell}\end{bmatrix}.
\]

From the recursive definition of the matrices $A_{\ell}$ and $B_{\ell}$, we get that 
$A_{\ell} - B_{\ell} = \begin{bmatrix} A_{\ell-1} - B_{\ell -1} & {\bf{0}}^{2^{\ell-2}} \\ B_{\ell-1} - A_{\ell -1} & {\bf 0}^{2^{\ell-2}} \end{bmatrix}$. 
From the induction hypothesis, we get,  
\[A_{\ell} - B_{\ell} = \begin{bmatrix} A_2 - B_2 & 0 & \cdots & 0 \\ B_2 - A_2 & 0 & \cdots & 0\\ &&\vdots& \\ A_2 - B_2 & 0 & \cdots & 0 \\ B_2 - A_2 & 0 & \cdots & 0 \end{bmatrix} = 
\begin{bmatrix} 0& 2 & 0 & \cdots & 0 \\ 0& -2 & 0 & \cdots & 0 \\ &&&\vdots& \\ 0& 2 & 0 & \cdots & 0 \\ 0& -2 & 0 & \cdots & 0 \end{bmatrix}
\]
and therefore for every $j \in [2^{\ell-1}]$, 
\[ \lvert x_{t, \ell,j}- y_{t, \ell, j} \rvert = \alpha_{t, \ell, 2} .\]
\end{proof}

We are now ready to prove \Cref{prop:bound-on-alphas}.
\begin{proof}[Proof of Proposition~\ref{prop:bound-on-alphas}]
\item
\paragraph{(\ref{prop:alpha-condition})} Follows from the definition of $\alpha_{t, i, i}$ in Algorithm~\ref{algo:gadget}. 
\item
\paragraph{(\ref{prop:bound-on-alpha-ti1})} $\alpha_{t,i,1} = 10^{f(t,i)} = 10^{(i-1)!\nu_t}$, where $\nu_t$ is the $t^{th}$ prime greater than $n^4$. Since $\nu_t$ is increasing in $t$, and for a fixed $t$, $\alpha_{t,i,1}$ is increasing in $i$,  $\max\limits_{t, i} \{ \alpha_{t,i,1} \} = \alpha_{n, d, 1}$ and 
$\min\limits_{t, i} \{ \alpha_{t,i,1} \} = \alpha_{1, 2, 1}$. We had noted earlier that from Prime Number Theorem, the $n^{th}$ prime greater than $n^4$ is at most $n^5$. Therefore, 
\[ 10^{n^4} < 10^{\nu_1} =  \alpha_{1, 2, 1}  \leq  \alpha_{t,i,1} \leq  \alpha_{n, d,1} = 10^{(d-1)! \nu_n} < 10^{d! n^5}. \]
\item
\paragraph{(\ref{prop:bound-on-alpha-tir})} From the definitions in Algorithm~\ref{algo:gadget}, for every $1 < r < i-1$, $\alpha_{t,i,r} = 10^{g(t,i,r)}$.
Note that $\max\limits_{t, i, r} \{ g(t,i,r)\} = g(n, d, d-1) \leq nd^2$ and therefore, 
\[ \alpha_{t,i,r} \leq \alpha_{n, d, d-1} = 10^{g(n,d,d-1)} \leq 
10^{nd^2}. \]
\item
\paragraph{(\ref{prop:bound-on-alpha-tii1} and \ref{prop:bound-on-alpha-tii})} 
Fix an arbitrary $t \in [n]$. We prove by induction on $i \in \{2, \cdots, d\}$ that, 
\[\lvert \alpha_{t,i,i} \rvert < 2  \text{ and } \displaystyle\sum\limits_{r=2}^{i}\lvert \alpha_{t,i,r} \rvert \leq 10^{\nu - nd}. \]

For the base case, $i=2$, $\alpha_{t,2,2}= \frac{ b_t^2 - a_t^2 }{2 \alpha_{t,2,1}}$. Recall from Definition~\ref{eqn:def-vars} that the variable part of $a_t$ and $b_t$ is the same. Therefore $\lvert b_t -a_t \rvert \leq 10^{m+\nu}$. 
From Fact~\ref{claim:bound-at}, we know that $\lvert a_t \rvert $ and $\lvert b_t \rvert$ are at most $10^{m+\nu+n+1} = 10^M$, so we get, 
\[\lvert b_t^2 - a_t^2 \rvert = \lvert (b_t-a_t)(b_t+a_t)\rvert  \leq  10^{m+\nu} \cdot 2 \max\{ a_t, b_t \} < 10^{m+\nu} \cdot 2\cdot 10^M. \]
Since $m+\nu < M$, $\lvert \alpha_{t,2,2} \rvert < \frac{ 10^{m+\nu} \cdot 2 \cdot 10^M }{2 \cdot 10^{f(t, 2)} } < 10^{2M-f(t,2)}$. By definitions in Algorithm~\ref{algo:gadget}, $f(t,2) = \nu_t $ and $\nu_t$ is a prime larger than $n^4$. Also, $M = O(n^3)$ and  $f(t,2)> 2M$ therefore, it follows that, $\lvert \alpha_{t,2,2} \rvert < 1 <10^{\nu-nd} $ and the claim holds for $i = 2$.

Let us assume the induction hypothesis for all $i < \ell \leq d$, and we now prove the claim for $i = \ell$. We need to show that 
\[ \lvert \alpha_{t,\ell,\ell} \rvert < 2  \text{ and }  \displaystyle\sum\limits_{r=2}^{\ell} \lvert \alpha_{t,\ell,r} \rvert \leq 10^{\nu-nd} \]
We first bound the magnitude of $\alpha_{t, \ell, \ell}$ for any $t\in [n]$. Recall from the definitions in Algorithm~\ref{algo:gadget}, 
\[ \lvert \alpha_{t, \ell, \ell} \rvert = \frac{ \lvert R_{t,\ell} \rvert }{i! \cdot \prod\limits_{r\in[\ell-1]} \alpha_{t,\ell,r}}, \text{ where, } R_{t, \ell} = b_t^{\ell} - a_t^{\ell} +  \displaystyle\sum\limits_{u=2}^{\ell-1} \displaystyle\sum\limits_{v=1}^{2^{u-1}} y_{t,u, v}^{\ell} - x_{t,u, v}^{\ell} .  \]
We will bound each individual term in the definition of $\alpha_{t, \ell, \ell}$ separately. 
\item The term $\lvert b_t^{\ell} - a_t^{\ell}\rvert$ in $R_{t, \ell}$ can be factorized as $\lvert b_t^{\ell} - a_t^{\ell} \rvert = \lvert (b_t - a_t)(\displaystyle\sum\limits_{k=0}^{\ell-1}b_t^ka_t^{\ell-1-k})\rvert$. We had seen earlier that $\lvert (b_t - a_t)\rvert < 10^{m+\nu} < 10^M$ and from Fact~\ref{claim:bound-at}, $\max\{a_t, b_t\} < 10^M$. Using these observations, we get
\begin{equation}
  \label{bound1}
  \begin{aligned}
\lvert b_t^{\ell} - a_t^{\ell} \rvert &= \lvert (b_t - a_t)(\displaystyle\sum\limits_{k=0}^{\ell-1}b_t^ka_t^{\ell-1-k}) \rvert \\
& < 10^M \cdot \ell \cdot \max\{a_t^{\ell-1}, b_t^{\ell-1}\} \\
& \leq 10^M \cdot \ell \cdot 10^{M(\ell-1)} = \ell \cdot 10^{M\ell}.
  \end{aligned}
\end{equation}

\item Using the definitions of $\alpha_{t,\ell,r}$, the denominator in the expression for $\alpha_{t,\ell,\ell}$ can be written as 
\begin{equation}
\label{bound2}
\begin{aligned}
\ell! \displaystyle\prod\limits_{r=1}^{\ell-1} \alpha_{t, \ell, r} &= \ell! \cdot 10^{f(t, \ell) + \sum\limits_{r=2}^{\ell-1} g(t, \ell, r)}\\
& \geq \ell! \cdot 10^{ f(t, \ell) + g(t, \ell, 2)}
\end{aligned}
\end{equation}

\item Now to bound the magnitude of $\left| \displaystyle\sum\limits_{u=2}^{\ell-1} \displaystyle\sum\limits_{v=1}^{2^{u-1}} y_{t,u, v}^{\ell} - x_{t,u, v}^{\ell} \right| $,  
we have
\begin{align*}
\left| \displaystyle\sum\limits_{u=2}^{\ell-1} \displaystyle\sum\limits_{v=1}^{2^{u-1}} y_{t,u, v}^{\ell} - x_{t,u, v}^{\ell} \right| 
&\leq \displaystyle\sum\limits_{u=2}^{\ell-1} \displaystyle\sum\limits_{v=1}^{2^{u-1}} \lvert y_{t,u, v}^{\ell} - x_{t,u, v}^{\ell} \rvert \\
&= \displaystyle\sum\limits_{u=2}^{\ell-1} \displaystyle\sum\limits_{v=1}^{2^{u-1}}\lvert (y_{t,u, v} - x_{t,u, v})(\displaystyle\sum\limits_{k=0}^{\ell-1} y_{t,u, v}^{k} x_{t,u, v}^{\ell-1-k}) \rvert \\
&\leq \displaystyle\sum\limits_{u=2}^{\ell-1} \displaystyle\sum\limits_{v=1}^{2^{u-1}}\lvert (y_{t,u, v} - x_{t,u, v}) \rvert \cdot \ell \cdot \max\{\lvert x_{t, u, v}\rvert ^{\ell-1}, \lvert y_{t, u, v}\rvert ^{\ell-1} \}
\end{align*}
Using Claim~\ref{claim:diff-x-y}, we know that for any $(t, u, v) \in [n] \times \{2, \cdots, \ell-1\} \times [2^{u-1}]$,
\[\lvert x_{t,u, v} - y_{t,u, v} \rvert = \alpha_{t,u,2} =10^{g(t, u, 2)}.\]
Also, from definition of the auxiliary variables in Algorithm~\ref{algo:gadget}, each $x_{t, u, v}$ and $y_{t, u, v}$ for any $t \in [n]$, is a $(\pm \frac12)$-linear combinations of $\{\alpha_{t,u,r} \mid r \in [u] \}$. Therefore,  
\[ \max\{ \lvert x_{t,u,v} \rvert , \lvert y_{t,u,v} \rvert \} \leq \frac12 \displaystyle\sum\limits_{r=1}^{u} \lvert \alpha_{t,u,r} \rvert.\]
Since $u < \ell$, using the induction hypothesis, we know that $\displaystyle\sum\limits_{r=2}^{u} \lvert \alpha_{t,u,r} \rvert < 10^{\nu-nd}$. So, we get
\[ \max\{ \lvert x_{t,u,v} \rvert , \lvert y_{t,u,v} \rvert \} \leq \lvert \alpha_{t, u, 1} \rvert  + \displaystyle\sum\limits_{r=2}^{u} \lvert \alpha_{t,u,r} \rvert  < \frac12 (10^{f(t,u)} + 10^{\nu-nd} )< 10^{f(t,u)}.\] 
From these observations, we get
\begin{align*}
\left| \displaystyle\sum\limits_{u=2}^{\ell-1} \displaystyle\sum\limits_{v=1}^{2^{u-1}} y_{t,u, v}^{\ell} - x_{t,u, v}^{\ell} \right| 
&\leq \displaystyle\sum\limits_{u=2}^{\ell-1} \displaystyle\sum\limits_{v=1}^{2^{u-1}}10^{g(t, u, 2)} \cdot 
\ell \cdot (10^{f(t,u)})^{\ell-1}
\end{align*}
Note that $\max\limits_{u}\{g(t, u, 2) \} =g(t, \ell-1, 2) $
and for a fixed $t, f(t,i)$ is increasing in $i$, therefore, $f(t, u) \leq f(t, \ell-1)$ for all $u \leq \ell-1$.  Therefore, 
\begin{equation}
\label{bound3}
\begin{aligned}
\left| \displaystyle\sum\limits_{u=2}^{\ell-1} \displaystyle\sum\limits_{v=1}^{2^{u-1}} y_{t,u, v}^{\ell} - x_{t,u, v}^{\ell} \right| 
 & \leq \ell \cdot 2^{\ell} \cdot 10^{g(t, \ell-1, 2)} \cdot 10^{(\ell-1)f(t,\ell-1)}
\end{aligned}
\end{equation}
Combining Equations~\ref{bound1}, \ref{bound2}, \ref{bound3}, we get an upper bound on the magnitude $\alpha_{t, \ell, \ell}$ as
\begin{align*}
\lvert \alpha_{t, \ell, \ell} \rvert&= \frac{\lvert R_{t,\ell} \rvert }{\ell! \displaystyle\prod\limits_{r=1}^{\ell-1} \alpha_{t, \ell, r}} \\
& \leq \frac{\lvert b_t^{\ell} - a_t^{\ell} \rvert}{\ell! \displaystyle\prod\limits_{r=1}^{\ell-1} \alpha_{t, \ell, r}} + 
\frac{\left| \displaystyle\sum\limits_{u=2}^{\ell-1} \displaystyle\sum\limits_{v=1}^{2^{u-1}} y_{t,u, v}^{\ell} - x_{t,u, v}^{\ell} \right|  }{\ell! \displaystyle\prod\limits_{r=1}^{\ell-1} \alpha_{t, \ell, r}}\\
&\leq  \frac{  \ell \cdot 10^{M \ell}} {\ell! \cdot 10^{ f(t, \ell) + g(t, \ell, 2)}} + \frac{\ell \cdot 2^{\ell} \cdot 10^{g(t, \ell-1, 2) +(\ell-1)f(t,\ell-1)}} {\ell! \cdot 10^{ f(t, \ell) + g(t, \ell, 2)}} 
\end{align*}
We now show that each individual term is at most $1$, and therefore, $\lvert \alpha_{t, \ell, \ell} \rvert < 2$.

\item The first term can be simplified by plugging in the definition of $f(t, \ell)$ and using the fact that $g(t,\ell, r) > 2$. 
\[ \frac{  \ell \cdot 10^{M \ell}} {\ell! \cdot 10^{ f(t, \ell) + g(t, \ell, 2)}}  < \frac{1}{(\ell-1)!} \cdot 10^{M \cdot \ell -(\ell-1)! \nu_t - 2}\]
Since $\ell \cdot M < (\ell-1)! \nu_t$ , it follows that 
\[ \frac{  \ell \cdot 10^{M \ell}} {\ell! \cdot 10^{ f(t, \ell) + g(t, \ell, 2)}}  < 1. \]

\item For the second term, note that $f(t, \ell) = (\ell-1)!\nu_t = (\ell-1) \cdot (\ell-2)! \nu_t = (\ell-1)f(t, \ell-1)$ and 
$g(t, \ell, 2) - g(t, \ell-1, 2) = 2\ell - 2 \geq 2$ for $\ell \geq 2$. Also, for $\ell \geq 2$, we have $\frac{2^\ell}{(\ell-1)!} \leq 4$. Therefore, 
\begin{align*} 
\frac{\ell \cdot 2^{\ell} \cdot 10^{g(t, \ell-1, 2) +(\ell-1)f(t,\ell-1)}} {\ell! \cdot 10^{ f(t, \ell) + g(t, \ell, 2)}} 
&= \frac{2^{\ell}}{(\ell-1)!} \cdot 10^{g(t, \ell-1, 2) - g(t, \ell, 2)} \cdot 10^{(\ell-1)f(t,\ell-1) - f(t, \ell) } \\
&\leq 4 \cdot 10^{-1} <1
\end{align*}

Now that we have established $\lvert \alpha_{t, \ell, \ell} \rvert < 2$, we show that  $\displaystyle\sum\limits_{r=2}^{\ell}\lvert \alpha_{t,\ell,r} \rvert < 10^{\nu-nd}$. We split this summation into two terms as 
\begin{align*}
\displaystyle\sum\limits_{r=2}^{\ell}\lvert \alpha_{t,\ell,r} \rvert = \displaystyle\sum\limits_{r=2}^{\ell-1} \lvert \alpha_{t,\ell,r} \rvert +  \lvert \alpha_{t,\ell,\ell} \rvert .
\end{align*}

\item From the definition of $\alpha_{t,\ell,r}$ for $1<r<\ell$, we have 
\[\displaystyle\sum\limits_{r=2}^{\ell-1} \lvert \alpha_{t,\ell,r} \rvert = \displaystyle\sum\limits_{r=2}^{\ell-1} 10^{g(t,\ell,r)} < 10^{g(t,\ell,\ell-1)+1}.\]
Since $g(t, i, r)$ is increasing in $t, i, r$,  $g(t,\ell,\ell-1)+1 \leq g(n ,d, d-1) +1 = nd^2$. Recall that $\nu = n^2$, and therefore, for any $d = o(\sqrt n)$, we have, 
\[ 10^{g(t,\ell,\ell-1)+1} \leq 10^{nd^2} \leq 10^{\nu- nd-1}.\]
Therefore, it follows that 
\[ \displaystyle\sum\limits_{r=2}^{\ell} \lvert \alpha_{t,\ell,r} \rvert \leq \displaystyle\sum\limits_{r=2}^{\ell-1} \lvert \alpha_{t,\ell,r} \rvert  + \vert \alpha_{t,\ell,\ell} \rvert <10^{\nu-nd-1} + 2 < 10^{\nu-nd}. \]
\end{proof}


\begin{proof}[Proof of Proposition~\ref{prop:bound-on-xtij}]

From the definition of $\alpha_{t, i, r}$ in Algorithm~\ref{algo:gadget}, we know that each auxiliary variable is a $(\pm \frac12)$-linear combination of $\{ \alpha_{t, i, r} \mid r \in [i] \}$. i.e
\[ x_{t,i,j} = \displaystyle\sum\limits_{r=1}^i u_r \alpha_{t,i,r}  ~\text{ for some } u_r \in \{ \pm \frac12 \}. \]
Therefore, 
\begin{align*}
\frac12 \cdot \lvert \alpha_{t, i, 1} \rvert -  \frac12 \cdot  \displaystyle\sum\limits_{r=2}^i \lvert \alpha_{t,i,r}  \rvert  
\leq & \left| \displaystyle\sum\limits_{r=1}^i u_r \alpha_{t,i,r} \right|
\leq \frac12 \cdot \lvert \alpha_{t, i, 1} \rvert +  \frac12 \cdot  \displaystyle\sum\limits_{r=2}^i \lvert \alpha_{t,i,r}  \rvert
\end{align*}
Using Proposition~\ref{prop:bound-on-alphas} (\ref{prop:bound-on-alpha-tii}), we know that $\displaystyle\sum\limits_{r=2}^i \lvert \alpha_{t,i,r} \rvert \leq 10^{\nu-nd}$ and from definitions, $\alpha_{t, i, 1} = 10^{(i-1)! \nu_t}$. Therefore, 
\[
\frac12 \cdot (10^{(i-1)! \nu_t} - 10^{\nu-nd} ) \leq \lvert x_{t,i,j} \rvert \leq \frac12 \cdot (10^{(i-1)! \nu_t} + 10^{\nu-nd})
\]
\end{proof}


\begin{proof}[Proof of Proposition~\ref{prop:all-distinct}]

Let $t_1, t_2 \in [n], i_1 , i_2 \in \{2, \cdots, d\}, j_1 \in [2^{i_1}-1]$ and $j_2 \in [2^{i_2}-1]$.
If $(t_1, i_1, j_1) = (t_2, i_2, j_2)$, then from Claim~\ref{claim:diff-x-y}, we know
$\lvert x_{t_1, i_1, j_1} - y_{t_1, i_1, j_1} \rvert = \alpha_{t,i,2} \neq 0$ and it follows that $x_{t_1, i_1, j_1} \neq y_{t_1, i_1, j_1}$. Now we show that if $(t_1, i_1, j_1) \neq (t_2, i_2, j_2)$, then $x_{t_1, i_1, j_1} \neq x_{t_2, i_2, j_2}$. 
The proof holds if either or both the $x_{t, i, j}$'s replaced with $y_{t, i, j}$. Let,   
\[x_{t_1, i_1, j_1} = u_1 \cdot 10^{(i_1-1)! \nu_{t_1}} + \displaystyle\sum\limits_{r=2}^{i_1}u_r \cdot \alpha_{t_1,i_1,r} ~\text{ for some } u_r \in \{ \pm \frac12 \} \]
and, 
\[x_{t_2, i_2, j_2}  = v_1 \cdot 10^{(i_2-1)! \nu_{t_2}} + \displaystyle\sum\limits_{r=2}^{i_2}v_r \cdot \alpha_{t_2,i_2,r} ~\text{ for some } v_r \in \{\pm \frac12 \}. \] 
If $x_{t_1, i_1, j_1} = x_{t_2, i_2, j_2}$, then on reordering the terms we get, 
\[\left| u_1 \cdot 10^{(i_1-1)! \nu_{t_1}} - v_1 \cdot 10^{(i_2-1)! \nu_{t_2}}\right| = \left| \displaystyle\sum\limits_{r=2}^{i_2}v_r \cdot \alpha_{t_2,i_2,r}  - \displaystyle\sum\limits_{r=2}^{i_1}u_r \cdot \alpha_{t_1,i_1,r} \right| \] 
Note that if $\lvert u_1 \cdot 10^{(i_1-1)! \nu_{t_1}} - v_1 \cdot 10^{(i_2-1)! \nu_{t_2}}\rvert $ is non-zero, then using the fact that $\nu_{t_1}$ and $\nu_{t_2}$ are prime integers larger than $n^4$ we have, 
 \[ \left| u_1 \cdot 10^{(i_1-1)! \nu_{t_1}} - v_1 \cdot 10^{(i_2-1)! \nu_{t_2}} \right|  \geq 10^{n^4} \]
But from Proposition~\ref{prop:bound-on-alphas}, (\ref{prop:bound-on-alpha-tii}),
\[ \left| \displaystyle\sum\limits_{r=2}^{i_2}v_r \cdot \alpha_{t_2,i_2,r}  - \displaystyle\sum\limits_{r=2}^{i_1}u_r \cdot \alpha_{t_1,i_1,r} \right| \leq  
\frac12 \displaystyle\sum\limits_{r=2}^{i_2} \lvert \alpha_{t_2,i_2,r} \rvert  + \frac12 \displaystyle\sum\limits_{r=2}^{i_1} \lvert \alpha_{t_1,i_1,r} \rvert
 \leq 10^{\nu - nd}\]
 which is a contradiction. Therefore,  $t_1 = t_2, i_1 = i_2 $ and $ u_1 = v_1$.

Let us assume $t_1 = t_2 = t, i_1 = i_2 = i$ and $j_1 > j_2$. If $x_{t, i, j_1} = x_{t, i, j_2}$, then, 
\begin{align*}
 \displaystyle\sum\limits_{r=2}^{i} (v_r - u_r) \cdot \alpha_{t,i,r} = 0
\end{align*} 
We know that $(v_r - u_r) \in \{0, \pm1 \}$, so there exists a $\{ 0, \pm 1\}$- linear combination of $\alpha_{t,i,r}$ equal to $0$. 
If $u_r = v_r$ for every $r \in \{2, \cdots, i\}$, then $j_1 = j_2$ since 
each auxiliary variable is a distinct linear combination of the $\alpha_{t, i, r}$'s. So, there exists at least one $r \in \{ 2, \cdots, i\}$ such that $u_r \neq v_r$. Let $r^*$ be the largest such $r$. We know that 
\[
0 = \left|  \displaystyle\sum\limits_{r=2}^{i} (v_r - u_r) \cdot \alpha_{t,i,r} \right| \geq \left| \lvert \alpha_{t, i, r^*} \rvert - \lvert \displaystyle\sum\limits_{r=2}^{r^*-1} (v_r - u_r) \cdot \alpha_{t,i,r}\rvert \right|
\]
But each $\alpha_{t,i,r} = 10^{g(t, i, r)}$ for $r \in \{2, \cdots, i-1\}$ is a distinct power of $10$ and $\lvert \alpha_{t,i,i} \rvert < 2$. So, $ \left|  \alpha_{t, i, r^*} \right| - \left| \displaystyle\sum\limits_{r=2}^{r^*-1} (v_r - u_r) \cdot \alpha_{t,i,r}\right| \neq  0 $, which is a contradiction.  Therefore, $j_1 = j_2$.  
\end{proof}

\section{Existence of (Inhomogeneous) PTE Solutions over General Finite Fields}
\label{subsec:existence}

Recall that a solution to a PTE system of size $s$ and degree $d$ satisfies 
\begin{eqnarray*}
x_1+x_2+\dots+x_s &=& y_1+y_2+\dots+ y_s\\
x_1^2+x_2^2+\dots+x_s^2 &=& y_1^2+y_2^2+\dots+ y_s^2\\
&\dots&\\
x_1^d+x_2^d+\dots+x_s^d &=& y_1^d+y_2^d+\dots+ y_s^d.
\end{eqnarray*}

We will show that such a system always has a solution over a field $\F=\F_{p^{\ell}}$,  for $d<|\F|^{1/2-\delta}$, for $\delta>0$. In fact, the proof will also hold for inhomogeneous PTE systems such as (\ref{eq:mom_match}).

\begin{theorem}\label{thm:existence} Let $\F$ be a finite field, and let $r_1, r_2, \ldots, r_d \in \F$. Let $d$ be a positive integer such that $d \le | \F|^{1/2-\delta}$. Then, there exists a solution in $\F$ to the system  $\sum_{i=1}^{s} x_i^j -\sum_{i=1}^{s} y_i^j=r_j$, for $j\in [d]$, with $s = 3d/\delta$.

Moreover, if $|\F|$ is a sufficiently large function of $\delta$, then we can ensure that the $x_i$'s and the $y_i$'s are all distinct.

\end{theorem}

Let $G$ be a group. An additive character of $G$ is a a function $\chi: G\ra \C$ such that $\chi(x+y)=\chi(x)\chi(y)$ for all $x, y\in G$. We will now define characters over groups of the form $\F^n$, where $\F=\F_{p^{\ell}}$ and $p$ is a prime.

 Let $\omega=e^{2\pi i/p}$ be a primitive $p$th root of unity, and let $Tr:\F_{p^{\ell}}\ra \F_p$ be the Trace operator $Tr(x)=\sum_{i=0}^{\ell-1} x^{p^i}$. Then, an additive character of $\F^n=(\F_{p^{\ell}})^n$ is  $\chi_a(x)=\omega^{Tr(a\cdot x)}$, where $a, x \in \F^n$, and $a\cdot x$ denotes the inner product over $\F^n$.

We will use of some results of \cite{KoppartyS13}. Let $\mu$ be a distribution over vectors in $\F^n$, and denote by $\mu^{(s)}$ the distribution of $x_1+x_2+  \ldots +x_s$, where the  $x_i$'s are picked independently from $\mu$.

\begin{theorem}(\cite{KoppartyS13}, Appendix B) \label{mu-s}Suppose that for some $\beta$,  any non-trivial character $\chi$ of $\F^s$ satisfies
$$|\E_{x\sim \mu} \chi(x)|\leq \beta.$$
Then $$\sum_{x\in \F^n}  \left| \mu^{(s)}(x)-\frac{1}{|\F|^n} \right| \leq \beta^s |\F|^n,$$  
and so $\mu^{(s)}$ is $\beta^s |\F|^n$- close to the uniform distribution over $\F^n$  in statistical distance.
\end{theorem}

Recall the Weil/Deligne bound.

\begin{theorem}(Weil \cite{weil}, Deligne \cite{deligne})\label{deligne}
Let $f(x_1, x_2, \ldots, x_t)$ be a $t$-variate polynomial over $\F$ of degree at most $|\F|^{1/2-\delta}$, for some $\delta>0$. Then, either $\chi(f(x))$ is constant for all $x\in \F$, or $\chi$ satisfies $|\E_{x\in \F}~ \chi(f(x))| \le |\F|^{-\delta}$.
\end{theorem}

\begin{proof}[Proof of Theorem \ref{thm:existence}]
For $x, y\in \F$, let $v_{x, y}=(x-y, x^2-y^2, \ldots, x^d-y^d)\in \F^d$.
Let $\mu$ be the distribution of $v_{x, y}$ when  $x, y$ are distributed independently and uniformly in $\F$. 
Note that for a nontrivial character $\chi_{a}$ with  $a\in (\F^*)^d$, we have 

$$\E_{v_{x,y} \sim \mu} ~[\chi_a(v_{x, y})]=\E[ \omega^{a \cdot v_{x, y}}]=\E_{x, y}[ \omega^{g(x, y)}]=\E_{x, y}[\chi_a(g(x, y))]$$ 
for the polynomial $g(x, y)=\sum_{i=1}^{d}~ a_i~ (x^i-y^i)$ of degree $d \le |\F|^{1/2-\delta}$.

By Deligne's Theorem \ref{deligne}, we have that 
\begin{equation}\label{beta}|\E_{v_{x,y} \sim \mu}~[ \chi_a(v_{x, y})]|=|\E[\chi_a (g(x, y))]| \le |\F|^{-\delta}.
\end{equation}

Let $\mu^{(s)}$ be the distribution of $S=\sum_{i=1}^s v_{x_i, y_i}$ when we pick $s$ vectors  $v_{x_1, y_1}, v_{x_2, y_2}, \ldots, v_{x_s, y_s}\in \F^d$ independently, according to $\mu$. Note that $\mu^{(s)}$ is precisely the distribution of 
$(\sum x_i-\sum y_i, \sum x_i^2-\sum y_i^2, \ldots, \sum x_i^d-\sum y_i^d)$, when we pick the $x_i$'s and $y_i$'s  independently and uniformly in $\F$.

By Theorem \ref{mu-s} and Equation (\ref{beta}), it follows that $$\sum_{v\in \F^d} \left| \mu^{(s)}(v)-\frac{1}{|\F|^d} \right| \leq (|\F|^{-\delta})^s |\F|^d=|\F|^{-\delta s+d}.$$

Picking $s=3d/\delta$, we get that $\mu^{(s)}((r_1, r_2, \ldots, r_d)) \geq |\F|^{-d}-|\F|^{-2d}>0$.

We can also ensure that all $x_i$'s and $y_i$'s are distinct, by noticing that the $\Pr[|\{x_1, x_2, \ldots, x_s, y_1, \ldots, y_s\}|=2s]=\prod_{i=0}^{2s-1} \frac{1}{|\F|-i}< (|\F|-2s)^{-2s}<{|\F|}^{-2d}< |\F|^{-d}-|\F|^{-2d}$ for $|\F|$ being sufficiently large as a function of $\delta$.

\end{proof}

\section{Reduction from 1-in-3 SAT to \mss{d} over $\F_{p^\ell}$}\label{subsec:general-fin-fields-reduction}
We will choose prime $p=O(d!)$ and  
$\ell = \poly(n)$ for this reduction.
To generate the field $\F_q = \F_{p^\ell}$, we consider an irreducible polynomial over $\F_p$ of degree $\ell$. 
Let $\gamma$ be a root of this polynomial in the algebraic closure of $\mathbb{F}_p$. Every element of $\F_q$ can then be generated as a linear combination of  $1, \gamma, \cdots, \gamma^{\ell-2}, \gamma^{\ell-1}$ over $\F_p$ (We refer to \cite{lidl97} for a general treatment of finite fields.).
Then, for $v=\sum v_i \gamma^i\in \F_q$, we will abuse notation and view $v$ as the vector $(v_1, v_2, \ldots, v_{\ell-1})$.
We define an analogue of the notion of ``magnitude'' used in the previous sections. For $v \in \F_q$, define  $\lvert v \rvert$ to be the largest non-zero index $i \in [\ell]$ in the vector representation of $v$.  
Note that this definition of magnitude satisfies the property that $\lvert u + v \rvert \le \max(\lvert u \rvert,\lvert v \rvert)$ for every $u, v \in \mathbb{F}_q$, and thus also satisfies that the triangle inequality. 

We now sketch a proof of the reduction, which follows analogously to the proof over the rational field, with some small modifications, as described next. 

An instance of \mss{d} consists of a tuple $\langle A, k, B_1, \dots, B_d \rangle$. 
Similar to the rational field reduction, each variable $(z_t, \overline{z_t})$ is mapped to $2^{d+1}-2$ distinct elements $\{a_{t}\} \cup \{ x_{t,i} \mid i \in [2^{d}-2] \}$ (corresponding to $z_t$) and $\{b_{t}\} \cup \{ y_{t,i} \mid i \in [2^{d}-2] \}$ (corresponding to $\overline{z}_t$). Let $\{a'_t, b'_t: t \in [n]\}$ be the elements of $\F_q$ produced by the reduction of 1-in-3 SAT to Subset-Sum defined as follows:
\begin{itemize}
\item The vector representations of $a'_t$ and $b'_t$ consist of two parts: a clause region consisting of the leftmost $m$ coordinates and a variable region consisting of the next  $n$ indices.
\item In the variable region, $a'_t$ and $b'_t$ have a $1$ at the $t$-th index and $0$'s at the other indices. Denote that by $(a_t)^{'v}$.
\item In the clause region, for every $j \in [m]$, $a'_t$ (resp. $b'_t$) has a $1$ at the $j$th location if $z_t$ (resp. $\overline{z_t}$) appears in clause $j$, and a $0$ otherwise. We denote the clause part of $a'_t$ by $(a_t)^{'c}$.
\item $ a_t' = ( a_t^{'c}, a_t^{'v}, 0^{\ell-m-n})$. Similarly for $b_t'$. 
\item The target $B$ is set to the element whose field representation is the vector which takes $1$'s in the first $m+n$ indices and 0 everywhere else. i.e. 
$B = (1^{m}, 1^n, 0^{\ell-m-n})$.
\end{itemize}

Define,
\[a_t = (0^{\nu}, a_t^{'c}, a_t^{'v}, 0^{\ell-\nu-m-n}) ~~\text{ and, }~~ b_t = (0^{\nu}, b_t^{'c}, b_t^{'v}, 0^{\ell-\nu-m-n}).\] 
For each $t \in [n]$, we will explicitly construct two sets of $2^{d}-2$ {\em auxiliary variables}, 
$X_t = \{ x_{t, i} \mid i \in [2^{d}-2] \}$ and $Y_t = \{ y_{t, i} \mid i \in  [2^{d}-2] \}$ which satisfy the following properties:
\begin{enumerate}[Property (1):]
\item\label{fp1} $\displaystyle\sum\limits_{x \in X_t} x = \displaystyle\sum\limits_{y \in Y_t} y = 0 $\\
\item\label{fp2} $\displaystyle\sum\limits_{x \in X_t} x^k - \displaystyle\sum\limits_{y \in Y_t} y^k  = b_t^k - a_t^k \text{ for every } k \in \{2,\ldots,d\}$.
\item\label{fp3} Additionally, an appropriately scaled set of auxiliary variables, can be shown to  satisfy the bimodal property. Namely, 
for any subset $S \subseteq  \bigcup\limits_{t\in[n]} X_t \cup Y_t$, and a scaling factor $K =\gamma^h $, where $h = \poly(n, d!)$,  
either
$$ \left| \displaystyle\sum\limits_{s \in S} \gamma^h s \right| >  h+n^4 
~~\text{or}~~ 
\left| \displaystyle\sum\limits_{s \in S} \gamma^h s \right| < h + \nu.$$
\end{enumerate}

Define the set $A = \displaystyle\bigcup\limits_{t\in[n]} \{ a_t \} \cup \{ b_t \} \cup X_t \cup Y_t . $
The targets $B_1, \dots, B_d$ are defined as follows:
\[B_1= ( 0^{\nu},  1^m, 1^n, 0^{\ell-\nu-m-n}),\] 
\[B_k = \displaystyle\sum\limits_{t = 1}^n a_t^k + \displaystyle\sum\limits_{t = 1}^n\displaystyle\sum\limits_{x \in X_t} x^k \text{   for every } k \in \{2, \dots, d \} \] 

Note that $a_t = \gamma^\nu \cdot a_t'$. Similarly, $b_t =\gamma^\nu \cdot b_t' $ and $B_1 =\gamma^\nu \cdot B$.

We now define a scaled version of the subset sum instance over the finite fields.  Let $h = \poly(n, d!)$ 
and let 
$K = \gamma^{h}$,  be the scaling factor 
. Scaling of all the elements of the instance of \mss{d} is roughly equivalent to scaling the rational solutions by a large power of $10$. The scaling of all the variables maintains Properties~\ref{fp1}, and \ref{fp2} of the auxiliary variables and also satisfies the solution to achieve Property~\ref{fp3}. 

Let $A_{h}= \{ \gamma^h a \mid a \in A \}$. 
$B_{k, h} = \gamma^{kh} B_k$    for every  $k \in \{2, \dots, d \}$. 
The following lemma shows that the \mss{d} instance and its scaled version as defined above are equivalent. 
\begin{lemma}
\label{lemma:scaling-ff}
Let $h = \poly(n, d!)$ 
and $K = \gamma^{h}$ be the scaling factor. There exists a subset $S \subseteq A$ such that for every $k \in [d]$
\[ \displaystyle\sum\limits_{s \in S} s^k = B_k .\]
if and only if there exists a subset $S_{h} \subseteq A_{h}$ such that for every $k \in [d]$
\[ \displaystyle\sum\limits_{s \in S_{h}} s^k = B_{k, h} .\]
\end{lemma}

The proof of Lemmas~\ref{lemma:scaling-ff} is straightforward and follows from the fact that the moment equations in \mss{d} are homogeneous and therefore scaling all the variables and the targets does not change the problem. 

We can then state the analogous statement of Lemma~\ref{lemma:reduction}, which implies the NP-hardness of \mss{d} over $\F_q$.
\begin{lemma}
\label{lemma:reduction-ff}
There exists a satisfying assignment to a 3-SAT instance $\phi(z_1, \dots, z_n)$ if and only if there exists a subset $S\subseteq A_{h}$ such that for every $k \in [d]$, 
\[ \displaystyle\sum\limits_{s \in S} s^k = B_{k, h} .\]
\end{lemma}

The proof of Lemma~\ref{lemma:reduction-ff} follows from the properties of the auxiliary variables stated above and all the steps of the proof over the rationals can be carried over here,  because we chose to scale the instance by a large enough power of $\gamma$, and we chose $p$ and $\ell$ large enough, in order to ensure that there is no wrapping around when we add terms with large magnitudes. 

\section{Conclusion}

The main open question that comes up from this work is to explicitly and efficiently construct degree-$d$ PTE solutions of size subexponential in $d$ (\Cref{main-question}). It would also be very interesting to prove analogous NP-hadness results for Bounded Distance Decoding of Reed-Solomon codes in the case where preprocessing is allowed. Finally, our NP-hardness results for Reed-Solomon codes apply to the case where the field size is exponential in the block length $N$; it would be very interesting to prove analogous NP-hardness results for smaller fields.
\section*{Acknowledgements}\label{sec:ack}

We would like to thank Madhu Sudan for very helpful discussions that led to the proof of existence of inhomogeneous PTE solutions over finite fields. We would also like to thank Venkatesan Guruswami and Swastik Kopparty for helpful comments and conversations. Finally, we would like to thank Andrew Sutherland and Colin Ingalls for helpful correspondence.

\bibliographystyle{alpha}
\bibliography{RS-NPHard}

\end{document}